\documentclass[12pt,english]{article}
\usepackage{lmodern}

\usepackage[T1]{fontenc}
\usepackage[a4paper]{geometry}
\usepackage{color}
\usepackage{float}
\usepackage{rotfloat}
\usepackage{textcomp}
\usepackage{amsmath}
\usepackage{amssymb}
\usepackage{graphicx}
\usepackage{setspace}
\makeatletter

\newcommand{\noun}[1]{\textsc{#1}}
\providecommand{\tabularnewline}{\\}
\floatstyle{ruled}
\newfloat{algorithm}{tbp}{loa}
\providecommand{\algorithmname}{Algorithm}
\floatname{algorithm}{\protect\algorithmname}

\usepackage{float}
\usepackage{amsfonts}
\usepackage{float}
\usepackage{amsfonts}
\usepackage{amsthm}
\usepackage{psfrag}
\usepackage{epsfig}
\usepackage[table]{xcolor}
\newtheorem{prop}{Proposition}[section]
\newtheorem{defin}{Definition}[section]
\date{}

\@ifundefined{showcaptionsetup}{}{%
 \PassOptionsToPackage{caption=false}{subfig}}
\usepackage{subfig}
\makeatother

\usepackage{babel}
\begin{document}

\title{Discriminative variable selection for clustering\\
 with the sparse Fisher-EM algorithm}

\author{Charles \noun{Bouveyron$^{*}$ \&} Camille \noun{Brunet$^{\dagger}$}}

\maketitle
\noindent \begin{center}
\vspace{-3ex}%
\begin{minipage}[t]{0.47\columnwidth}%
\begin{center}
\noun{$^{*}$ }{\small Laboratoire SAMM, EA 4543}\\
{\small{} Université Paris 1 Panthéon-Sorbonne}
\par\end{center}%
\end{minipage}%
\begin{minipage}[t]{0.47\columnwidth}%
\begin{center}
\noun{$^{\dagger}$}{\small{} Equipe Modal'X, EA 3454}\\
{\small Université Paris Ouest Nanterre}
\par\end{center}%
\end{minipage}\vspace{3ex}
\par\end{center}
\begin{abstract}
The interest in variable selection for clustering has increased recently
due to the growing need in clustering high-dimensional data. Variable
selection allows in particular to ease both the clustering and the
interpretation of the results. Existing approaches have demonstrated
the efficiency of variable selection for clustering but turn out to
be either very time consuming or not sparse enough in high-dimensional
spaces. This work proposes to perform a selection of the discriminative
variables by introducing sparsity in the loading matrix of the Fisher-EM
algorithm. This clustering method has been recently proposed for the
simultaneous visualization and clustering of high-dimensional data.
It is based on a latent mixture model which fits the data into a low-dimensional
discriminative subspace. Three different approaches are proposed in
this work to introduce sparsity in the orientation matrix of the discriminative
subspace through $\ell_{1}$-type penalizations. Experimental comparisons
with existing approaches on simulated and real-world data sets demonstrate
the interest of the proposed methodology. An application to the segmentation
of hyperspectral images of the planet Mars is also presented.
\end{abstract}

\section{Introduction}

\begin{sloppypar}With the exponential growth of measurement capacities,
the observed data are nowadays frequently high-dimensional and clustering
such data remains a challenging problem. In particular, when considering
the mixture model context, the corresponding clustering methods show
a disappointing behavior in high-dimensional spaces. They suffer from
the well-known \textit{curse of dimensionality}~\cite{Bellman57}
which is mainly due to the fact that model-based clustering methods
are dramatically over-parametrized in high-dimensional spaces. Moreover,
even though we dispose of many variables to describe the\textit{\textcolor{red}{{}
}}studied phenomenon, most of the time, only a small subset of these
original variables are in\textit{\textcolor{red}{{} }}fact relevant.\end{sloppypar}

Several recent works have been interested to simultaneously cluster
data and reduce their dimensionality by selecting relevant variables
for the clustering task. A common assumption to these works is that
the true underlying clusters are assumed to differ only with respect
to some of the original features. The clustering task aims therefore
to group the data on a subset of relevant features. This presents
two practical advantages: clustering results should be improved by
the removing of non informative features and the interpretation of
the obtained clusters should be eased by the meaning of retained variables.
In the literature, variable selection for clustering is handled in
two different ways.

On the one hand, some authors such as~\cite{Law04,Liu03,Maugis09,Raftery06}
tackle the problem of variable selection for model-based clustering
within a Bayesian framework. In particular, the determination of the
role of each variable is recast as a model selection problem. A first
framework was proposed by Raftery and Dean~\cite{Raftery06} in which
two kinds of subsets of variables are defined: a subset of relevant
variables and a subset of irrelevant variables which are independent
from the clustering but which can be explained from the relevant variables
through a linear regression. An extension of the previous work has
then been proposed by Maugis \emph{et al.}~\cite{Maugis09} who consider
two kinds of irrelevant variables: the ones which can be explained
by a linear regression from a subset of the clustering variables and
finally a set of irrelevant variables which are totally independent
of all the relevant variables. The models in competition are afterward
compared with the integrated log-likelihood \emph{via} a BIC approximation.
Even though these approaches present good results in most practical
situations, their computational times are nevertheless very high and
can lead to an intractable procedure in the case of high-dimensional
data.

On the other hand, penalized clustering criteria have also been proposed
to deal with the problem of variable selection in clustering. In the
Gaussian mixture model context, several works, such as \cite{Pan07,Wang08,Xie08,Zhang09}
in particular, introduced a penalty term in the log-likelihood function
in order to yield sparsity in the features. The penalty function can
take different forms according to the constraints imposed on the structure
of the covariance matrices. The introduction of a penalty term in
the log-likelihood function was also used in the mixture of factor
analyzers approaches, such as in~\cite{Galimberti2009,Xie10}. More
recently, Witten and Tibshirani~\cite{Witten10} proposed a general
non-probabilistic framework for variable selection in clustering,
based on a general penalized criterion, which governs both variable
selection and clustering. It appears nevertheless that the results
of such procedures are usually not sparse enough and select a large
number of the original variables, especially in the case of high-dimensional
data.

Other approaches focus on simultaneously clustering the data and reducing
their dimensionality by feature extraction rather than feature selection.
We can cite in particular, the subspace clustering methods \cite{Bouveyron07b,Ghahramani97,McNicholas2008,McLachlan2003,Montanari2010,Yoshida04}
which are based on probabilistic frameworks and model each group in
a specific and low-dimensional subspace. Even though these methods
are very efficient in practice, they present nevertheless several
limitations regarding the understanding and the interpretation of
the clusters. Indeed, in most of subspace clustering approaches, each
group is modeled in its specific subspace which makes difficult a
global visualization of the clustered data. Even though some approaches~\cite{Baek2009,Montanari2010}
model the data in a common and low-dimensional subspace, they choose
the projection matrix such as the variance of the projected data is
maximum and this can not be sufficient to catch discriminative information
about the group structure.

To overcome these limitations, Bouveyron and Brunet~\cite{Bouveyron12a}
recently proposed a new statistical framework which aims to simultaneously
cluster the data and produce a low-dimensional and discriminative
representation of the clustered data. The resulting clustering method,
named the Fisher-EM algorithm, clusters the data into a common latent
subspace of low dimensionality which best discriminates the groups
according to the current fuzzy partition of the data. It is based
on an EM procedure from which an additional step, named F-step, is
introduced to estimate the projection matrix whose columns span the
discriminative latent space. This projection matrix is estimated at
each iteration by maximizing a constrained Fisher's criterion conditionally
to the current soft partition of the data. As reported in~\cite{Bouveyron12a},
the Fisher-EM algorithm turned out to outperform most of the existing
clustering methods while providing a useful visualization of the clustered
data. However, the discriminative latent space is defined by {}``latent
variables'' which are linear combinations of the original variables.
As a consequence, the interpretation of the resulting clusters according
to the original variables is usually difficult. An intuitive way to
avoid such a limitation would be to keep only large loadings variables,
by thresholding for instance. Even though this approach is commonly
used in practice, it has been particularly criticized by Cadima~\cite{Cadima1995}
since it induces some misleading information. Furthermore, it often
happens when dealing with high-dimensional data that a large number
of noisy or non-informative variables are present in the set of the
original variables. Since the latent variables are defined by a linear
combination of the original ones, the noisy variables may remain in
the loadings of the projection matrix and this may produce a deterioration
of the clustering results. 

To overcome these shortcomings, three different approaches are proposed
in this work for introducing sparsity in the Fisher-EM algorithm and
thus select the discriminative variables among the set of original
variables. The remainder of this document is organized as follows.
Section~2 reviews the discriminative latent mixture model of~\cite{Bouveyron12a}
and the Fisher-EM algorithm which was proposed for its inference.
Section~3 develops three different procedures based on $\ell_{1}$
penalties for introducing sparsity into the Fisher-EM algorithm. The
first approach looks for the best sparse approximate of the solution
of the F-step of the Fisher-EM algorithm. The second one recasts the
optimization problem involved of the F-step as a lasso regression-type
problem. The last approach is based on a penalized singular value
decomposition (SVD) of the matrix involved in the constrained Fisher
criterion of the F-step. Numerical experiments are then presented
in Section~4 to highlight the practical behavior of the three sparse
versions of the Fisher-EM algorithm and to compare them to existing
approaches. In section~5, a sparse version of the Fisher-EM algorithm
is applied to the segmentation of hyperspectral images. Section~6
finally provides some concluding remarks and ideas for further works.

\section{The DLM model and the Fisher-EM algorithm\label{sec:The-DLM-model}}

In this section, we briefly review the discriminative latent mixture
(DLM) model~\cite{Bouveyron12a} and its inference algorithm, named
the Fisher-EM algorithm, which models and clusters the data into a
common latent subspace. Conversely to similar approaches, such as~\cite{Bouveyron07,McNicholas2008,Montanari06,Montanari2010,Yoshida04},
this latent subspace is assumed to be discriminative and its intrinsic
dimension is strictly bounded by the number of groups.

\subsection{The DLM model}

\begin{figure}
\begin{centering}
\psfrag{X}{\hspace{-1ex}$X$} \psfrag{Y}{\hspace{-1ex}$Y$}
\psfrag{Z}{\hspace{-1ex}$Z$} \psfrag{pi}{\hspace{-1ex}$\pi$}
\psfrag{mu}{\hspace{0ex}$\mu_k$} \psfrag{sig}{\hspace{0ex}$\Sigma_k$}
\psfrag{V}{\hspace{0ex}$W=[U,V]$} \psfrag{e}{\hspace{-1ex}$\epsilon$}
\psfrag{psi}{\hspace{0ex}$\Psi$}\includegraphics[bb=-5bp -5bp 380bp 283bp,clip,width=0.5\columnwidth]{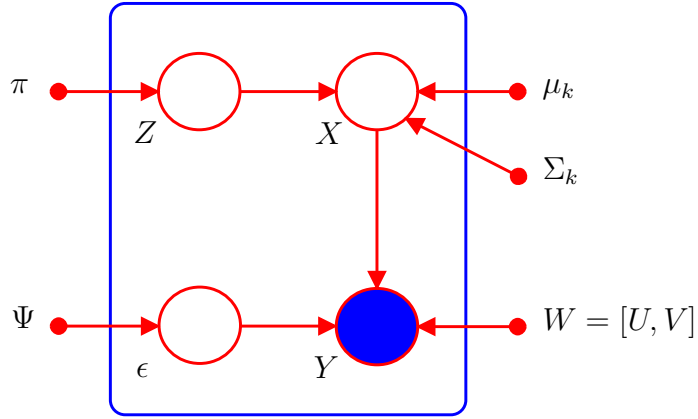}
\par\end{centering}

\caption{Graphical summary of the DLM$_{[\Sigma_{k}\beta]}$ model.\label{fig:Graphical_model}}
\end{figure}

Let $\{y_{1},\dots,y_{n}\}\in\mathbb{R}^{p}$ denote a dataset of
$n$ observations that one wants to cluster into $K$ homogeneous
groups, \emph{i.e.} adjoin to each observation $y_{i}$ a value $z_{i}\in\{1,\dots,K\}$
where $z_{i}=k$ indicates that the observation $y_{i}$ belongs to
the $k$th group. On the one hand, let us assume that $\{y_{1},\dots,y_{n}\}$
are independent observed realizations of a random vector $Y\in\mathbb{R}^{p}$
and that $\{z_{1},\dots,z_{n}\}$ are also independent realizations
of a random variable $Z\in\{1,\dots,K\}$. On the other hand, let
$\mathbb{E}\subset\mathbb{R}^{p}$ denote a latent space assumed to
be the most discriminative subspace of dimension $d\leq K-1$ such
that $\mathbf{0}\in\mathbb{E}$ and $K<p$. Moreover, let $\{x_{1},\dots,x_{n}\}\in\mathbb{E}$
denote the actual data, described in the latent space~$\mathbb{E}$
of dimension $d$, which are in addition presumed to be independent
realizations of an unobserved random vector $X\in\mathbb{E}$. Finally,
the observed variable $Y\in\mathbb{R}^{p}$ and the latent variable
$X\in\mathbb{E}$ are assumed to be linked through a linear transformation:

\begin{equation}
Y=UX+\varepsilon,\label{eq:linear_relationship}
\end{equation}
where $U$ is a $p\times d$ orthonormal matrix common to the $K$
groups and satisfying $U^{t}U=\mathbf{I}_{d}$. The $p$-dimensional
random vector $\varepsilon$ stands for the noise term which models
the non discriminative information and which is assumed to be distributed
according to a centered Gaussian density function with a covariance
matrix $\Psi$ ($\varepsilon\sim\mathcal{N}(0,\Psi)$). Besides, within
the latent space, $X$ is assumed, conditionally to $Z=k$, to be
Gaussian : 
\begin{equation}
X_{|Z=k}\sim\mathcal{N}(\mu_{k},\Sigma_{k})
\end{equation}
where $\mu_{k}\in\mathbb{R}^{d}$ and $\mbox{\ensuremath{\Sigma}}_{k}\in\mathbb{R}^{d\times d}$
are respectively the mean vector and the covariance matrix of the
$k$th group. Given these distribution assumptions and according to
equation~(\ref{eq:linear_relationship}), 
\begin{equation}
Y_{|X,Z=k}\sim\mathcal{N}(UX,\Psi),\label{eq:Y}
\end{equation}
and its marginal distribution is therefore a mixture of Gaussians:
\begin{equation}
f(y)=\sum_{k=1}^{K}\pi_{k}\phi(y;m_{k},S_{k}),\label{eq:GMM}
\end{equation}
where $\pi_{k}$ is the mixing proportion of the $k$th group and
$\phi(.;m_{k},S_{k})$ denotes the multivariate Gaussian density function
parametrized by the mean vector $m_{k}=U\mu_{k}$ and the covariance
matrix $S_{k}=U\Sigma_{k}U^{t}+\Psi$ of the $k$th group. Furthermore,
we define the $p\times p$ matrix $W=[U,V]$ such that $W^{t}W=WW^{t}=\mathbf{I}_{p}$,
where the $(p-d)\times p$ matrix $V$ is an orthogonal complement
of $U$. Finally, the noise covariance matrix $\Psi$ is assumed to
satisfy the conditions $V\Psi V^{t}=\beta\mathbf{I}_{p-d}$ and $U\Psi U^{t}=\mathbf{0}_{d}$,
such that $\Delta_{k}=W^{t}S_{k}W$ has the following form: $$\Delta_k=\left(  \begin{array}{c@{}c} \begin{array}{|ccc|}\hline ~~ & ~~ & ~~ \\  & \Sigma_k &  \\  & & \\ \hline \end{array} & \mathbf{0}\\ \mathbf{0} &  \begin{array}{|cccc|}\hline \beta & & & 0\\ & \ddots & &\\  & & \ddots &\\ 0 & & & \beta\\ \hline \end{array} \end{array}\right)  \begin{array}{cc} \left.\begin{array}{c} \\\\\\\end{array}\right\}  & d \leq K-1\vspace{1.5ex}\\ \left.\begin{array}{c} \\\\\\\\\end{array}\right\}  & (p-d)\end{array}$$These
last conditions imply that the discriminative and the non-discriminative
subspaces are orthogonal, which suggests in practice that all the
relevant clustering information remains in the latent subspace. This
model is referred to by DLM$_{[\Sigma_{k}\beta]}$ in~\cite{Bouveyron12a}
and a graphical summary is given in Figure~\ref{fig:Graphical_model}.

\subsection{A family of parsimonious models}

Several other models can be obtained from the DLM$_{[\Sigma_{k}\beta]}$
model by relaxing or adding constraints on model parameters. Firstly,
it is possible to consider a more general case than the DLM$_{[\Sigma_{k}\beta]}$
by relaxing the constraint on the variance term of the non discriminative
information. Assuming that $\varepsilon_{|Z=k}\sim\mathcal{N}(0,\Psi_{k})$
yields the DLM$_{[\Sigma_{k}\beta_{k}]}$ model which can be useful
in some practical cases. From this extended model, 10 parsimonious
models can be obtained by constraining the parameters $\Sigma_{k}$
and $\beta_{k}$ to be common between and within the groups. For instance,
the covariance matrices $\Sigma_{1},\dots,\Sigma_{K}$ in the latent
space can be assumed to be common across the groups and this sub-model
is referred to by $\mathrm{DLM}_{[\Sigma\beta_{k}]}$. Similarly,
in each group, $\Sigma_{k}$ can be assumed to be diagonal, \emph{i.e.}
$\Sigma_{k}=\mathrm{diag}(\alpha_{k1},\dots,\alpha_{kd})$. This sub-model
is referred to by $\mathrm{DLM}_{[\alpha_{kj}\beta_{k}]}$. These
sub-models can also be declined by considering that the parameter
$\beta$ is common to all classes ($\forall k,\,\beta_{k}=\beta$).
A list of the $12$ different DLM models is given by Table~\ref{Tab:models}
and detailed descriptions can be found in~\cite{Bouveyron12a}. Such
a family yields very parsimonious models and allows, in the same time,
to fit into various situations. In particular, the complexity of the
DLM$_{[\Sigma_{k}\beta_{k}]}$ model mainly depends on the number
of clusters $K$ since the dimensionality of the discriminative subspace
is such that $d\leq K-1$. Notice that the complexity of the DLM$_{[\Sigma_{k}\beta_{k}]}$
grows linearly with $p$ contrary to the traditional Gaussian models
in which the complexity increases with~$p^{2}$. As an illustration,
if we consider the case where $p=100$, $K=4$ and $d=3$, then the
number of parameters to estimate for the DLM$_{[\Sigma_{k}\beta_{k}]}$
is $337$ which is drastically less than in the case of the Full-GMM
($20\,603$ parameters to estimate). For a comparison purpose, Table~\ref{Tab:models}
presents also the complexity of other clustering methods, such as
Mixt-PPCA~\cite{Tipping99b}, MFA~\cite{McLachlan2003}, PGMM~\cite{McNicholas2008},
MCFA~\cite{Baek09} and MCUFSA~\cite{Yoshida06} for which the complexity
grows linearly with $p$ as well.

\begin{table}[t]
\begin{centering}
{\footnotesize\onehalfspacing%
\begin{tabular}{lll}
Model  & Nb. of parameters  & %
\begin{tabular}{l}
$K=4$ and\tabularnewline
$p=100$\tabularnewline
\end{tabular}\tabularnewline
\hline 
$\mathrm{DLM}_{[\Sigma_{k}\beta_{k}]}$ & $(K-1)+K(K-1)+(K-1)(p-K/2)+K^{2}(K-1)/2+K$ & 337\tabularnewline
$\mathrm{DLM}_{[\Sigma_{k}\beta]}$ & $(K-1)+K(K-1)+(K-1)(p-K/2)+K^{2}(K-1)/2+1$ & 334\tabularnewline
$\mathrm{DLM}_{[\Sigma\beta_{k}]}$ & $(K-1)+K(K-1)+(K-1)(p-K/2)+K(K-1)/2+K$ & 319\tabularnewline
$\mathrm{DLM}_{[\Sigma\beta]}$ & $(K-1)+K(K-1)+(K-1)(p-K/2)+K(K-1)/2+1$ & 316\tabularnewline
$\mathrm{DLM}_{[\alpha_{kj}\beta_{k}]}$ & $(K-1)+K(K-1)+(K-1)(p-K/2)+K^{2}$  & 325\tabularnewline
$\mathrm{DLM}_{[\alpha_{kj}\beta]}$ & $(K-1)+K(K-1)+(K-1)(p-K/2)+K(K-1)+1$  & 322\tabularnewline
$\mathrm{DLM}_{[\alpha_{k}\beta_{k}]}$ & $(K-1)+K(K-1)+(K-1)(p-K/2)+2K$  & 317\tabularnewline
$\mathrm{DLM}_{[\alpha_{k}\beta]}$ & $(K-1)+K(K-1)+(K-1)(p-K/2)+K+1$  & 314\tabularnewline
$\mathrm{DLM}_{[\alpha_{j}\beta_{k}]}$ & $(K-1)+K(K-1)+(K-1)(p-K/2)+(K-1)+K$  & 316\tabularnewline
$\mathrm{DLM}_{[\alpha_{j}\beta]}$ & $(K-1)+K(K-1)+(K-1)(p-K/2)+(K-1)+1$  & 313\tabularnewline
$\mathrm{DLM}_{[\alpha\beta_{k}]}$ & $(K-1)+K(K-1)+(K-1)(p-K/2)+K+1$  & 314\tabularnewline
$\mathrm{DLM}_{[\alpha\beta]}$ & $(K-1)+K(K-1)+(K-1)(p-K/2)+2$ & 311\tabularnewline
\hline 
Full-GMM  & $(K-1)+Kp+Kp(p+1)/2$  & 20603 \tabularnewline
Com-GMM  & $(K-1)+Kp+p(p+1)/2$  & 5453\tabularnewline
Diag-GMM & $(K-1)+Kp+Kp$  & 803 \tabularnewline
Sphe-GMM & $(K-1)+Kp+K$  & 407 \tabularnewline
\hline 
MFA & $(K-1)+Kp+Kd[p-(d-1)/2)]+Kp$ & 1991 ($d=3$)\tabularnewline
Mixt-PPCA  & $(K-1)+Kp+K[d(p-(d+1)/2)+d+1]+1$ & 1198 ($d=3$)\tabularnewline
PGMM-CUU & $(K-1)+Kp+d[p-(d+1)/2]+Kp$ & 1100 ($d=3$)\tabularnewline
MCFA  & $(K-1)+Kd+p+d[p-(d+1)/2]+Kd(d+1)/2$ & 433\textcolor{white}{0} ($d=3$)\tabularnewline
MCUFSA  & $(K-1)+Kd+1+d[p-(d+1)/2]+Kd$ & 322\textcolor{white}{0} ($d=3$)\tabularnewline
\hline 
\end{tabular}}
\par\end{centering}

\caption{\label{Tab:models} Number of free parameters to estimate when $d=K-1$
for the DLM models and some classical models (see text for details).}
\end{table}

\subsection{The Fisher-EM algorithm}

An estimation procedure, called the Fisher-EM algorithm, is also proposed
in~\cite{Bouveyron12a} in order to estimate both the discriminative
space and the parameters of the mixture model. This algorithm is based
on the EM algorithm from which an additional step is introduced, between
the E and the M-step. This additional step, named F-step, aims to
compute the projection matrix $U$ whose columns span the discriminative
latent space. The Fisher-EM algorithm has therefore the following
form, at iteration~$q$:

\paragraph*{The E-step}

This step computes the posterior probabilities $t_{ik}^{(q)}$ that
the observations belong to the $K$ groups using the following update
formula:
\begin{equation}
t_{ik}^{(q)}=\hat{\pi}_{k}^{(q-1)}\phi(y_{i},\hat{\theta}_{k}^{(q-1)})/\sum_{\ell=1}^{K}\hat{\pi}_{\ell}^{(q-1)}\phi(y_{i},\hat{\theta}_{\ell}^{(q-1)}),\label{eq:t_ik}
\end{equation}
with $\hat{\theta}_{k}=\{\hat{\mu}_{k},\hat{\Sigma}_{k},\hat{\beta}_{k},\hat{U}\}$.

\paragraph*{The F-step}

This step estimates, conditionally to the posterior probabilities,
the orientation matrix $U^{(q)}$ of the discriminative latent space
by maximizing the Fisher's criterion~\cite{Fisher36,Fukunaga90}
under orthonormality constraints:
\begin{eqnarray}
\hat{U}^{(q)} & = & \max_{U}\quad\mathrm{trace}\left((U^{t}SU)^{-1}U^{t}S_{B}^{(q)}U\right),\nonumber \\
 &  & \text{w.r.t.}\quad U^{t}U=\mathbf{I}_{d},\label{eq:Opti_Fisher}
\end{eqnarray}
where $S$ stands for the covariance matrix of the whole dataset and
$S_{B}^{(q)}$, defined as follows: 
\begin{equation}
S_{B}^{(q)}=\frac{1}{n}\sum_{k=1}^{K}n_{k}^{(q)}(m_{k}^{(q)}-\bar{y})(m_{k}^{(q)}-\bar{y})^{t},\label{eq:Sb_soft}
\end{equation}
denotes the soft between covariance matrix with $n_{k}^{(q)}=\sum_{i=1}^{n}t_{ik}^{(q)}$,
$m_{k}^{(q)}=1/n_{k}^{(q)}\sum_{i=1}^{n}t_{ik}^{(q)}y_{i}$ and $\bar{y}=1/n\sum_{i=1}^{n}y_{i}$.
This optimization problem is solved in~\cite{Bouveyron12a} using
the concept of orthonormal discriminant vector developed by~\cite{Foley75}
through a Gram-Schmidt procedure. Such a process enables to fit a
discriminative and low-dimensional subspace conditionally to the current
soft partition of the data while providing orthonormal discriminative
axes. In addition, according to the rank of the matrix $S_{B}^{(q)}$,
the dimensionality of the discriminative space $d$ is strictly bounded
by the number of clusters~$K$.

\paragraph*{The M-step}

This third step estimates the parameters of the mixture model in the
latent subspace by maximizing the conditional expectation of the complete
log-likelihood:{\small 
\begin{align}
Q(\theta) & =-\frac{1}{2}\sum_{k=1}^{K}n_{k}^{(q)}\Bigl[\text{-}2\log(\pi_{k})+\mathrm{trace}(\Sigma_{k}^{-1}\hat{U}^{(q)t}C_{k}^{(q)}\hat{U}^{(q)})+\log(\left|\Sigma_{k}\right|)\nonumber \\
 & \,\,\,\,\,\,\,+(p\text{-}d)\log(\beta_{k})+\frac{\mathrm{trace}(C_{k}^{(q)})\text{-}\sum_{j=1}^{d}\hat{u}_{j}^{(q)t}C_{k}^{(q)}\hat{u}_{j}^{(q)}}{\beta_{k}}+p\log(2\pi)\Bigr].
\end{align}
}where $C_{k}^{(q)}=\frac{1}{n_{k}^{(q)}}\sum_{i=1}^{n}t_{ik}^{(q)}(y_{i}-m_{k}^{(q)})(y_{i}-m_{k}^{(q)})^{t}$
is the empirical covariance matrix of the $k$th group and $\hat{u}_{j}^{(q)}$
is the $j$th column vector of $\hat{U}^{(q)}$, $n_{k}^{(q)}=\sum_{i=1}^{n}t_{ik}^{(q)}$.
Hence, maximizing $Q$ conditionally to $\hat{U}^{(q)}$ leads to
the following update formula for the mixture parameters of the model
DLM$_{[\Sigma_{k}\beta_{k}]}$:
\begin{align}
\hat{\pi}_{k}^{(q)} & =\frac{n_{k}^{(q)}}{n},\\
\hat{\mu}_{k}^{(q)} & =\frac{1}{n_{k}^{(q)}}\sum_{i=1}^{n}t_{ik}^{(q)}\hat{U}^{(q)t}y_{i},\\
\hat{\Sigma}_{k}^{(q)} & =\hat{U}^{(q)t}C_{k}\hat{U}^{(q)},\\
\hat{\beta}_{k}^{(q)} & =\frac{\mathrm{trace}(C_{k})\text{-}\sum_{j=1}^{d}\hat{u}_{j}^{(q)t}C_{k}\hat{u}_{j}^{(q)}}{p-d}.
\end{align}
The Fisher-EM procedure iteratively updates the parameters until the
Aitken criterion is satisfied (see paragraph~4.5 of~\cite{Bouveyron12a}).
The convergence properties of the Fisher-EM algorithm have been studied
in~\cite{Bouveyron12}. It is also proposed in this work to use a
stopping criterion based on the Fisher criterion involved in the F-step
to improve the clustering performance. Finally, since the latent subspace
has a low dimension and common to all groups, the clustered data can
be easily visualized by projecting them into the estimated latent
subspace.

\section{Sparse versions of the Fisher-EM algorithm}

Even though the Fisher-EM algorithm turns out to be efficient both
for modeling and clustering data, the interpretation of clustering
results regarding the original variables remains difficult. In this
section, we propose therefore three different ways to introduce sparsity
into the loadings of the projection matrix estimated in the F-step
of the Fisher-EM algorithm.

\subsection{A two-step approach}

In this first approach, we propose to proceed in two steps. First,
at iteration $q$, the traditional F-step of the Fisher-EM algorithm
computes an estimate $\hat{U}^{(q)}$ of the orientation matrix of
the discriminative latent space conditionally to the posterior probabilities
$t_{ik}^{(q)}$. Then, the matrix $\hat{U}^{(q)}$ is approximated
by a sparse one $\tilde{U}^{(q)}$ using the following result.

\begin{prop}\label{Prop-1}The best sparse approximation $\tilde{U}^{(q)}$
of $\hat{U}^{(q)}$ at the level $\lambda$ is the solution of the
following penalized regression problem: \textit{
\[
\min_{\mathcal{U}}\left\Vert X^{(q)t}-Y^{t}\mathcal{U}\right\Vert _{F}^{2}+\lambda\sum_{j=1}^{d}\left|\mathcal{U}_{j}\right|_{1},
\]
}where $\mathcal{U}=[\mathcal{U}_{1},...,\mathcal{U}_{d}]$, $\mathcal{U}_{j}\in\mathbb{R}^{p}$
is the $j$th column vector of\emph{ }$\mathcal{U}$, $\left\Vert .\right\Vert _{F}$
is the Frobenius norm and $X^{(q)}=\hat{U}^{(q)t}Y$\textit{\emph{.}}\end{prop}

\begin{proof}Let $\hat{U}^{(q)}$ be the orientation matrix of the
discriminative latent space estimated by the F-step at iteration $(q)$
and let us define $X^{(q)}=\hat{U}^{(q)t}Y\,\in\mathbb{R}^{d\times n}$
the matrix of the projected data into the subspace spanned by $\hat{U}^{(q)}$,
where $Y\in\mathbb{R}^{p\times n}$ denotes the original data matrix.
Since $X^{(q)}$ is generated by $\hat{U}^{(q)}$, then $\hat{U}^{(q)}$
is solution of the least square regression of $X^{(q)}$ on $Y$:
\[
\min_{\mathcal{U}}\left\Vert X^{(q)t}-Y^{t}\mathcal{U}\right\Vert _{F}^{2},
\]
\textit{\emph{where $\mathcal{U}=[\mathcal{U}_{1},...,\mathcal{U}_{d}]$,
$\mathcal{U}_{j}\in\mathbb{R}^{p}$ is the $j$th column vector of}}\textit{
}\textit{\emph{$\mathcal{U}$, $\left\Vert .\right\Vert _{F}$ is
the Frobenius norm.}} A penalized version of this regression problem
can be obtained by adding a $\ell_{1}$-penalty term as follows:
\[
\min_{\mathcal{U}}\left\Vert X^{(q)t}-Y^{t}\mathcal{U}\right\Vert _{F}^{2}+\lambda\sum_{j=1}^{d}\left|\mathcal{U}_{j}\right|_{1},
\]
and the solution of this penalized regression problem is therefore
the best sparse approximation of $\hat{U}^{(q)}$ at the level $\lambda$.\end{proof}

The previous result allows to provide a sparse approximation $\tilde{U}^{(q)}$
of $\hat{U}^{(q)}$ but we have no guarantee that the $\tilde{U}^{(q)}$
is orthogonal as required by the DLM model. The following proposition
solves this issue.

\begin{prop}\label{Prop-2}The best orthogonal approximation of $\tilde{U}^{(q)}$
is $\bar{U}^{(q)}=u^{(q)}v^{(q)t}$ where $u^{(q)}$ and $v^{(q)}$
are respectively the left and right singular vectors of the SVD of
$\tilde{U}^{(q)}$.\end{prop}

\begin{proof}Let us consider the matrix $\tilde{U}^{(q)}$. Searching
the best orthogonal approximation of the matrix $\tilde{U}^{(q)}$
is equivalent to solving the following optimization problem:
\[
\min_{\mathcal{U}}\left\Vert \tilde{U}^{(q)}-\mathcal{U}\right\Vert _{F}^{2}\text{\,\,\ w.r.t. }\mathcal{U}^{t}\mathcal{U}=\mathbf{I_{d}}.
\]
This problem is a nearest orthogonal Procrustes problem which can
be solved by a singular value decomposition~\cite{Gower04}. Let
$u^{(q)}\Lambda^{(q)}v^{(q)t}$ be the singular value decomposition
of $\tilde{U}^{(q)}$, then $u^{(q)}v^{(q)t}$ is the best orthogonal
approximation of $\tilde{U}^{(q)}$. \end{proof}

From an practical point of view, the penalized regression problem
of Proposition~\ref{Prop-1} can be solved by alternatively regressing
each column vector of the projected matrix $\hat{U}^{(q)}$. The sparse
and orthogonal approximation $\bar{U}^{(q)}$ of $\tilde{U}^{(q)}$
is obtained afterward through a SVD of $\tilde{U}^{(q)}$. The following
algorithm summarizes these steps.

\begin{algorithm}[H]
\caption{-- F-step of the sparseFEM-1 algorithm}

\begin{enumerate}
\item At iteration $q$, compute the matrix $\hat{U}^{(q)}$ by solving~(\ref{eq:Opti_Fisher}).
\item Compute $X^{(q)}=\hat{U}^{(q)t}Y$.
\item For $j\in\{1,\dots,d\}$, solve $d$ independent penalized regression
problems with the LARS algorithm~\cite{Efron2004}:
\[
\tilde{U}_{j}^{(q)}=\mathrm{arg}\min_{\mathcal{U}_{j}}\left\Vert x_{j}^{(q)t}-Y^{t}\mathcal{U}_{j}\right\Vert ^{2}+\lambda\left|\mathcal{U}_{j}\right|_{1},
\]

\item Repeat step 3 several times until convergence.
\item Let $\tilde{U}^{(q)}=[\tilde{U}_{1}^{(q)},...,\tilde{U}_{d}^{(q)}]$,
compute the SVD of $\tilde{U}^{(q)}=u{}^{(q)}\Lambda{}^{(q)}v{}^{(q)t}$
and let $\bar{U}^{(q)}=u{}^{(q)}v{}^{(q)t}.$\end{enumerate}
\end{algorithm}

Let us remark that this problem can be extended to a more general
penalized regression by adding a ridge penalty term. This allows in
particular to handle the $n<p$ case which occurs frequently nowadays.
In such a case, the elastic-net algorithm~\cite{Zou03} has to be
used instead of the LARS algorithm in Algorithm~1. 

Nevertheless, a limitation of such a procedure may be the disconnection
between the estimation of the discriminative subspace and the introduction
of the sparsity in the loadings of the projection matrix. To avoid
that, the two following approaches aim to propose penalized Fisher
criteria for which the solutions fit directly a sparse and discriminative
latent subspace.

\subsection{A penalized regression criterion}

We therefore propose here to reformulate the constrained Fisher criterion~(\ref{eq:Opti_Fisher})
involved in the F-step of the Fisher-EM algorithm as a penalized regression
problem. Consequently, the solution of this penalized regression problem
will fit directly a sparse and discriminative latent subspace. To
this end, let us introduce the soft matrices $H_{W}^{(q)}$ and $H_{B}^{(q)}$
which will be computed, conditionally to the E-step, at each iteration
$q$ of the sparse F-step as follows:

\begin{defin} The soft matrices $H_{W}^{(q)}\in\mathbb{R}^{p\times n}$
and $H_{B}^{(q)}\in\mathbb{R}^{p\times K}$ are defined, conditionally
to the posterior probabilities $t_{ik}^{(q)}$ computed in the E-step
at iteration $q$, \textcolor{black}{as follows:}

\vspace{-3ex}

\begin{eqnarray}
H_{W}^{(q)} & = & \frac{1}{\sqrt{n}}\left[Y-\sum_{k=1}^{K}t_{1k}^{(q)}m_{k}^{(q)},\dots,Y-\sum_{k=1}^{K}t_{nk}^{(q)}m_{k}^{(q)}\right]\in\mathbb{R}^{p\times n}\label{eq:soft_Hw}\\
H_{B}^{(q)} & = & \frac{1}{\sqrt{n}}\left[\sqrt{n_{1}^{(q)}}(m_{1}^{(q)}-\bar{y}),\dots,\sqrt{n_{K}^{(q)}}(m_{K}^{(q)}-\bar{y})\right]\in\mathbb{R}^{p\times K},\label{eq:soft_Hb}
\end{eqnarray}
where $n_{k}^{(q)}=\sum_{i=1}^{n}t_{ik}^{(q)}$ and $m_{k}^{(q)}=\frac{1}{n}\sum_{i=1}^{n}t_{ik}^{(q)}y_{i}$
is the soft mean vector of the cluster~$k$. \end{defin} 

\noindent According to these definitions, the matrices $H_{W}^{(q)}$
and $H_{B}^{(q)}$ satisfy:
\begin{equation}
H_{W}^{(q)}H_{W}^{(q)t}=S_{W}^{(q)}\quad\text{and}\quad H_{B}^{(q)}H_{B}^{(q)t}=S_{B}^{(q)},\label{eq:Hb}
\end{equation}
where $S_{W}^{(q)}=1/n\sum_{k=1}^{K}n_{k}^{(q)}C_{k}$ stands for
the soft within covariance matrix computed at iteration~$q$ and
$S_{B}^{(q)}$ denotes the soft between covariance matrix defined
in equation~(\ref{eq:Sb_soft}). A penalized version of the optimization
problem~(\ref{eq:Opti_Fisher}) can be therefore formulated as a
penalized regression-type problem:

\begin{prop}\label{Prop-3}The best sparse approximation $\tilde{U}^{(q)}$
of the solution of~(\ref{eq:Opti_Fisher}) at the level~$\lambda$
is the solution $\hat{B}^{(q)}$ of the following penalized regression
problem:
\begin{align*}
 & \min_{A,B}\sum_{k=1}^{K}\left\Vert R_{W}^{(q)-t}H_{B,k}^{(q)}-AB^{t}H_{B,k}^{(q)}\right\Vert _{F}^{2}+\rho\sum_{j=1}^{d}\beta_{j}^{t}S_{W}^{(q)}\beta_{j}+\lambda\sum_{j=1}^{d}\left|\beta_{j}\right|_{1},\\
 & \text{w.r.t. }A^{t}A=\mathbf{I}_{d},
\end{align*}
where\textit{\emph{ }}\emph{$A=[\alpha_{1},\dots,\alpha_{d}]\in\mathbb{R}^{p\times d}$},
\emph{$B=[\beta_{1},\dots,\beta_{d}]\in\mathbb{R}^{p\times d}$, }\textit{\emph{$R_{W}^{(q)}\in\mathbb{R}^{p\times p}$}}
is a upper triangular matrix resulting from the Cholesky decomposition
of $S_{W}^{(q)}$, \emph{i.e. $S_{W}^{(q)}=R_{W}^{(q)t}R_{W}^{(q)}$},
$H_{B,k}^{(q)}$ is the $k$th column of $H_{B}^{(q)}$ and \emph{$\rho>0$.}\end{prop}

\begin{proof}First, let us consider that the matrix $A$ is fixed
at iteration $q$ . Then, optimizing :
\begin{equation}
\min_{A,B}\sum_{k=1}^{K}\left\Vert R_{W}^{(q)-t}H_{B,k}^{(q)}-AB^{t}H_{B,k}^{(q)}\right\Vert _{F}^{2}+\rho\sum_{j=1}^{d}\beta_{j}^{t}S_{W}^{(q)}\beta_{j}\label{eq:regression_Fisher}
\end{equation}
conditionally to $A$ leads to consider the following regularized
regression problem:
\[
\min_{B}\sum_{j=1}^{d}\left[\left\Vert H_{B}^{(q)t}R_{W}^{(q)-t}\alpha_{j}-H_{B}^{(q)t}\beta_{j}\right\Vert _{F}^{2}+\rho\beta_{j}^{t}S_{W}^{(q)}\beta_{j}\right],
\]
with \textit{\emph{$B=[\beta_{1},\dots,\beta_{d}]$}}. Solving this
problem is equivalent to solving $d$ independent ridge regression
problem and the solution $\hat{B}^{(q)}$ is :
\begin{equation}
\hat{B}^{(q)}=\left(S_{B}^{(q)}+\rho S_{W}^{(q)}\right)^{-1}S_{B}^{(q)}R_{W}^{(q)-1}A.\label{eq:Bhat}
\end{equation}
By substituting $\hat{B}^{(q)}$ in Equation~(\ref{eq:regression_Fisher}),
optimizing the objective function~(\ref{eq:regression_Fisher}) over
$A$, given $A^{t}A=\mathbf{I}_{d}$ and $\hat{B}^{(q)}$ fixed, is
equivalent to maximize the quantity:
\begin{align*}
 & \max_{A}\mathrm{trace}\left(\hat{B}^{(q)t}H_{B}^{(q)}H_{B}^{(q)t}R_{W}^{(q)-1}A\right),\\
 & \text{w.r.t. }A^{t}A=\mathbf{I}_{d}.
\end{align*}
According to Lemma~1 of~\cite{Qiao09}, this is a Procrustes problem~\cite{Gower04}
which has an analytical solution by computing the singular value decomposition
of the quantity:
\[
R_{W}^{(q)-t}(H_{B}^{(q)}H_{B}^{(q)t})\hat{B}^{(q)}=u^{(q)}\Lambda^{(q)}v^{(q)t},
\]
where the column vectors of the $p\times d$ matrix $u^{(q)}$ are
orthogonal and $v^{(q)}$ is a $d\times d$ orthogonal matrix. The
solution is $\hat{A}^{(q)}=u^{(q)}v^{(q)t}$. Substituting $\hat{A}^{(q)}$
into~\eqref{eq:Bhat} gives:
\begin{eqnarray*}
\hat{B}^{(q)} & = & R_{W}^{(q)-1}\left(R_{W}^{(q)-t}S_{B}^{(q)}R_{W}^{(q)-1}+\rho\mathbf{I}_{p}\right)^{-1}R_{W}^{(q)-t}S_{B}^{(q)}R_{W}^{(q)-1}\hat{A}^{(q)}\\
 & = & R_{W}^{(q)-1}u^{(q)}\left(\Lambda^{(q)}+\rho\mathbf{I}_{p}\right)^{-1}\Lambda^{(q)}v^{(q)t}.
\end{eqnarray*}
By remarking that the $d$ eigenvectors associated to the non-zero
eigenvalues of the generalized eigenvalue problem~\eqref{eq:Opti_Fisher}
are the columns of $R_{W}^{(q)-1}u^{(q)},$ it follows that $\hat{B}^{(q)}$
spans the same linear subspace than the solution $\hat{U}^{(q)}$
of~(\ref{eq:Opti_Fisher}). Therefore, the solution of the penalized
optimization problem:
\begin{align*}
 & \min_{A,B}\sum_{k=1}^{K}\left\Vert R_{W}^{(q)-t}H_{B,k}^{(q)}-AB^{t}H_{B,k}^{(q)}\right\Vert _{F}^{2}+\rho\sum_{j=1}^{d}\beta_{j}^{t}S_{W}^{(q)}\beta_{j}+\lambda\sum_{j=1}^{d}\left|\beta_{j}\right|_{1},\\
 & \text{w.r.t. }A^{t}A=\mathbf{I}_{d},
\end{align*}
is the best sparse approximation of the solution of~(\ref{eq:Opti_Fisher})
at the level $\lambda$.\end{proof} 

However and as in the previous case, the orthogonality of the column
vectors of $\hat{B}^{(q)}$ is not guaranteed but this issue can be
tackled by Proposition~\ref{Prop-2}. \textcolor{black}{From a} practical
point of view, the optimization problem of Proposition~\ref{Prop-3}
can be solved using the algorithm proposed by~\cite{Qiao09} in the
supervised case by optimizing alternatively over $B$ with $A$ fixed
and over $A$ with $B$ fixed. This leads to the following algorithm
in our case: 

\begin{algorithm}[H]
\caption{-- F-step of the sparseFEM-2 algorithm}

\begin{enumerate}
\item At iteration $q$, compute the matrices $H_{B}^{(q)}$ and $H_{W}^{(q)}$
from Equations~\eqref{eq:soft_Hw} and~\eqref{eq:soft_Hb}. Let
$S_{W}^{(q)}=H_{W}^{(q)}H_{W}^{(q)t}$ and $S_{B}^{(q)}=H_{B}^{(q)}H_{B}^{(q)t}$.
\item Compute $R_{W}^{(q)}$ by using a Cholesky decomposition of $S_{W}^{(q)}+\gamma/p\,\mathrm{trace}(S_{W}^{(q)})=R_{W}^{(q)t}R_{W}^{(q)}$.
\item Initialization:

Let $B^{(q)}$ be the eigenvectors of $S^{-1}S_{B}^{(q)}$ .

Compute the SVD $R_{W}^{(q)-t}S_{B}^{(q)}B^{(q)}=u^{(q)}\Lambda^{(q)}v^{(q)t}$
and let $A^{(q)}=u^{(q)}v^{(q)t}$.

\item Solve $d$ independent penalized regression problems. For $j=1,\dots,d$:

\[
\hat{\beta}_{j}^{(q)}=\mathrm{arg}\min_{\beta_{j}}\left(\beta_{j}^{t}W^{(q)t}W^{(q)}\beta_{j}-2\tilde{Y}^{(q)t}W^{(q)}\beta_{j}+\lambda_{1}\left\Vert \beta_{j}\right\Vert _{1}\right),
\]
where $W^{(q)}=\left(\begin{array}{c}
H_{B}^{(q)t}\\
\sqrt{\rho}R_{W}^{(q)}
\end{array}\right)$ and $\tilde{Y}^{(q)}=\left(\begin{array}{c}
H_{B}^{(q)t}R_{W}^{(q)-1}\alpha_{j}^{(q)}\\
\mathbf{O}_{p}
\end{array}\right)$. 

\item Let $\hat{B}^{(q)}=[\hat{\beta}_{1},\dots,\hat{\beta}_{d}]$. Compute
the SVD of $R_{W}^{(q)-t}S_{B}^{(q)}\hat{B}^{(q)}=u^{(q)}\Lambda^{(q)}v^{(q)t}$
and let $A^{(q)}=u^{(q)}v^{(q)t}$.
\item Compute the SVD of $\hat{B}^{(q)}=u'^{(q)}\Lambda'^{(q)}v'^{(q)t}$
and let $\bar{U}^{(q)}=u'^{(q)}v'^{(q)t}.$
\item Repeat steps several times until convergence.\end{enumerate}
\end{algorithm}

\subsection{A penalized singular value decomposition}

In this last approach, we reformulate the constrained Fisher criterion~(\ref{eq:Opti_Fisher})
involved in the F-step of the Fisher-EM algorithm as a regression
problem\textcolor{black}{{} which can be solved by doing the SVD of
the matrix of interest in this regression problem. A sparse approximation
of the solution of this regression problem will be obtained by doing
a penalized SVD~\cite{Witten09} instead of the SVD. To that end,
}let us consider the following result.

\begin{prop}\label{prop-4}The solution of~\eqref{eq:Opti_Fisher}
is also solution of the following constrained optimization problem:\vspace{-2ex}
\begin{alignat*}{1}
 & \min_{\mathcal{U}}\sum_{\ell=1}^{p}\left\Vert S_{B,\ell}^{(q)}-\mathcal{U}\mathcal{U}^{t}S_{B,\ell}^{(q)}\right\Vert ^{2}\\
 & \text{w.r.t. }\mathcal{U}^{t}\mathcal{U}=\mathbf{I}_{d},
\end{alignat*}
where $S_{B,\ell}^{(q)}$ is the $\ell$th column of the soft between
covariance matrix $S_{B}^{(q)}$ computed at iteration $q$.\end{prop}\vspace{-3ex}

\begin{proof} Let us first prove that minimizing the quantity $\sum_{\ell=1}^{p}||S_{B,\ell}^{(q)}-UU^{t}S_{B,\ell}^{(q)}||^{2}$
is equivalent to maximize $\mathrm{trace}(U^{t}S_{B}^{(q)}S_{B}^{(q)t}U)$.
To that end, we can write down the following equalities:
\begin{eqnarray*}
\sum_{\ell=1}^{p}\left\Vert S_{B,\ell}^{(q)}-UU^{t}S_{B,\ell}^{(q)}\right\Vert ^{2} & = & \sum_{\ell=1}^{p}\mathrm{trace}\left(S_{B,\ell}^{(q)t}(\mathbf{I}_{p}-UU^{t})^{t}(\mathbf{I}_{p}-UU^{t})S_{B,\ell}^{(q)}\right)\\
 & = & \mathrm{trace}\left((\mathbf{I}_{p}-UU^{t})^{t}(\mathbf{I}_{p}-UU^{t})\sum_{\ell=1}^{p}S_{B,\ell}^{(q)}S_{B,\ell}^{(q)t}\right)\\
 & = & \mathrm{trace}\left(S_{B}^{(q)t}(\mathbf{I}_{p}-UU^{t})S_{B}^{(q)}\right)\\
 & = & \mathrm{trace}(S_{B}^{(q)t}S_{B}^{(q)})-\mathrm{trace}(U^{t}S_{B}^{(q)}S_{B}^{(q)t}U).
\end{eqnarray*}
Consequently, minimizing over $U$ the quantity $\sum_{\ell=1}^{p}||S_{B,\ell}^{(q)}-UU^{t}S_{B,\ell}^{(q)}||^{2}$
is equivalent to maximize $\mathrm{trace}(U^{t}S_{B}^{(q)}S_{B}^{(q)t}U)$
according to $U$. Let us now consider the SVD of the $n\times p$
matrix $S_{B}^{(q)}=u\Lambda v^{t}$ where $u$ and $v$ stands for
respectively the left and right singular vectors of $S_{B}^{(q)}$
and $\Lambda$ is a diagonal matrix containing its associated singular
values. Since the matrix $S_{B}^{(q)}$ has a rank $d$ at most equal
to $K-1<p$, with $K$ the number of clusters, then only $d$ singular
values of the matrix $S_{B}^{(q)}$ are non zeros, which enables us
to write $S_{B}^{(q)}=u\Lambda_{d}v^{t}$, where $\Lambda_{d}=\mathrm{diag}(\lambda_{1},\dots,\lambda_{d},0,\dots,0)$.
Moreover, by letting $U=u_{d}$ the $d$ first left eigenvectors of
$S_{B}$, then: 
\begin{eqnarray*}
\mathrm{trace}\left(U^{t}S_{B}S_{B}^{t}U\right) & = & \mathrm{trace}\left(U^{t}(u\Lambda_{d}v^{t})(u\Lambda_{d}v^{t})^{t}U\right),\\
 & = & \mathrm{trace}\left(U^{t}u\Lambda_{d}\Lambda_{d}^{t}u^{t}U\right),\\
 & = & \sum_{j=1}^{d}\lambda_{j}^{2}.
\end{eqnarray*}
Consequently, the $p\times d$ orthogonal matrix $\hat{U}$ such that
$\sum_{\ell=1}^{p}||S_{B,\ell}^{(q)}-UU^{t}S_{B,\ell}^{(q)}||^{2}$
is minimized, is the matrix made of the $d$ first left eigenvectors
of $S_{B}^{(q)}$. Besides, since $S_{B}^{(q)}$ is symmetric and
semi-definite positive, the matrix $\hat{U}$ contains also the eigenvectors
associated with the $d$ largest eigenvalues of $S_{B}^{(q)2}$ and
therefore the ones of $S_{B}^{(q)}$. Therefore, assuming without
loss of generality that $S=\mathbf{I}_{p}$, $\hat{U}$ is also solution
of the constrained optimization problem~(\ref{eq:Opti_Fisher}) involved
in the original F-step. \end{proof}

The optimization problem of Proposition~\eqref{prop-4} can be seen
as looking for the projection matrix $\mathcal{U}$ such that the
back-projection $\mathcal{U}\mathcal{U}^{t}S_{B,\ell}^{(q)}$ is as
close as possible to $S_{B,\ell}^{(q)}$. In~\cite{Witten09}, Witten
\emph{et al.} have considered such a problem with a constraint of
sparsity on $\mathcal{U}$. To solve this problem, they proposed an
algorithm which performs a penalized SVD of the matrix of interest
in the constrained optimization problem. Therefore, it is possible
to obtain a sparse approximation $\tilde{U}^{(q)}$ of the solution
of~\eqref{eq:Opti_Fisher} by doing a penalized SVD of $S_{B}^{(q)}$
with the algorithm of~\cite{Witten09}. As previously, the orthogonality
of the column vectors of $\tilde{U}^{(q)}$ is not guaranteed but
this issue can be again tackled by Proposition~\ref{Prop-2}. \textcolor{black}{From
a} practical point of view, this third approach can be implemented
as follows:

\begin{algorithm}[H]
{\small \caption{-- F-step of the sparseFEM-3 algorithm}
}{\small \par}
\begin{enumerate}
\item Let $M_{1}=S_{B}^{(q)}$ and $d=\mathrm{rank}\left(S_{B}\right).$
\item For $j\in\{1,\dots,d\}$:

\begin{enumerate}
\item Solve $\hat{u}_{j}^{(q)}=\mathrm{arg}\max_{u_{j}}u_{j}^{t}M_{j}v_{j}$
w.r.t. $\left\Vert u_{j}\right\Vert _{2}^{2}\leq1$, $\left\Vert v_{j}\right\Vert _{2}^{2}\leq1$
and $\sum_{\ell=1}^{p}\left|u_{j\ell}^{(q)}\right|\leq\lambda_{1}$
using the penalized SVD algorithm of~\cite{Witten09}.
\item Update $M_{j+1}=M_{j}-\lambda_{j}u_{j}^{(q)}v_{j}^{t}$.
\end{enumerate}
\item $\hat{U}^{(q)}=[\hat{u}_{1}^{(q)},\dots,\hat{u}_{d}^{(q)}]$.
\item Compute the SVD of $\hat{U}^{(q)}=u^{(q)}\Lambda^{(q)}v^{(q)t}$ and
let\textit{\emph{ $\tilde{U}^{(q)}=u^{(q)}v^{(q)t}$.}}\end{enumerate}
\end{algorithm}

\subsection{Practical aspects}

The introduction of sparsity in the Fisher-EM algorithm presents several
practical aspects among which the ability to interpret the discriminative
axes. However, two questions remain: the choice of the hyper-parameter
which determines the level of sparsity and the implementation strategy
in the Fisher-EM algorithm. Both aspects are discussed below.

\paragraph*{Choice of the tuning parameter}

The choice of the threshold $\lambda$ is an important problem since
the number of zeros in the $d$ discriminative axes depends directly
on the degree of sparsity. In~\cite{Zou2006}, Zou \emph{et al.}
chose the hyper-parameter of their sparse PCA with a criterion based
on the explanation of the variance approximated by the sparse principal
components. In~\cite{Witten10}, Witten and Tibshirani proposed for
their sparse-kmeans to base the choice of the tuning parameter on
a permutation method closely related to the gap statistic previously
proposed by Tibshirani \emph{et al.}~\cite{Tibshirani01} for estimating
the number of components in standard kmeans. Since our model is defined
in a Gaussian mixture context, we propose to use the BIC criterion
to select the threshold $\lambda$. According to the consistency results
obtained by~Zou \emph{et al.}~\cite{Zou07} and the fact that the
sparsity constraint is applied on the projection matrix $U$, the
effective number of parameters to estimate in the DLM$_{[\Sigma_{k}\beta_{k}]}$
model is:
\[
\gamma_{e}=(K-1)+Kd+\left(d[p\text{−}(d+1)/2]-\mathbf{d_{e}}\right)+Kd(d+1)/2+K
\]
where $\mathbf{d_{e}}$ is the number of zeros in the loading matrix.
In the same manner, this effective number of parameters to estimate
can be declined for the $11$ other sub-models of the DLM family.

\paragraph*{Implementation of the sparse Fisher-EM algorithm}

We identified two different ways to implement the sparse versions
of the Fisher-EM algorithm. First, it could be possible to replace
the usual F-step of the Fisher-EM algorithm by a sparse F-step developed
previously. The resulting algorithm would sparsify at each iteration
the projection matrix $U$ before estimating the model parameters.
This can however leads to some drawbacks since an early introduction
of the $\ell_{1}$ penalty could penalize too much the loadings of
the projection matrix, in particular if the initialization is far
away from the optimal situation. An alternative implementation would
be to, first, execute the traditional Fisher-EM algorithm until convergence
and, then, initialize the sparse Fisher-EM algorithm with the result
of the Fisher-EM algorithm. This strategy should combine the efficiency
of the standard Fisher-EM algorithm with the advantage of having a
sparse selection of discriminative variables. We therefore recommend
this second implementation and it will be used in the experiments
presented in the following sections.

\section{Experimental comparison}

This section presents comparisons with existing variable selection
techniques on simulated and real-world data sets.

\subsection{Comparison on simulated data}

\begin{table}[p]
\begin{centering}
\begin{tabular}{|c|l|c|c|}
\hline 
Simulation & Method & Clustering error & non-zero variables\tabularnewline
\hline 
\hline 
$n=30$ $\mu=0.6$ & kmeans & $0.48\pm0.05$ & $25.0\pm0.0$\tabularnewline
 & sparse-kmeans & $0.47\pm0.07$ & $19.0\pm6.6$\tabularnewline
 & Clustvarsel & $0.62\pm0.06$ & $22.2\pm1.2$\tabularnewline
 & Selvarclust & $0.40\pm0.03^{*}$ & $8.1\pm1.9^{*}$\tabularnewline
\cline{2-4} 
 & sparseFEM-1 & $0.47\pm0.06$ & $2.6\pm0.9$\tabularnewline
 & sparseFEM-2 & $0.48\pm0.07$ & $4.7\pm1.8$\tabularnewline
 & sparseFEM-3 & $0.48\pm0.03$ & $2.0\pm0.0$\tabularnewline
\hline 
\hline 
$n=30$ $\mu=1.7$ & kmeans & $0.14\pm10.2$ & $25.0\pm0.0$\tabularnewline
 & sparse-kmeans & $0.08\pm0.06$ & $23.6\pm0.8$\tabularnewline
 & Clustvarsel & $0.41\pm0.10$ & $16.6\pm10.4$\tabularnewline
 & Selvarclust & $0.08\pm0.08^{*}$ & $6.8\pm1.4^{*}$\tabularnewline
\cline{2-4} 
 & sparseFEM-1 & $0.14\pm0.13$ & $3.5\pm0.8$\tabularnewline
 & sparseFEM-2 & $0.20\pm0.12$ & $5.4\pm2.2$\tabularnewline
 & sparseFEM-3 & $0.17\pm0.11$ & $2.0\pm0.0$\tabularnewline
\hline 
\hline 
$n=300$ $\mu=0.6$ & kmeans & $0.43\pm0.03$ & $25.0\pm0.0$\tabularnewline
 & sparse-kmeans & $0.46\pm0.03$ & $24.0\pm0.5$\tabularnewline
 & Clustvarsel & $0.42\pm0.03$ & $25.0\pm0.0$\tabularnewline
 & Selvarclust & $0.34\pm0.02^{*}$ & $7.0\pm1.7^{\text{*}}$\tabularnewline
\cline{2-4} 
 & sparseFEM-1 & $0.42\pm0.03$ & $2.4\pm1.0$\tabularnewline
 & sparseFEM-2 & $0.43\pm0.03$ & $5.2\pm2.7$\tabularnewline
 & sparseFEM-3 & $0.43\pm0.04$ & $2.3\pm1.1$\tabularnewline
\hline 
\hline 
$n=300$ $\mu=1.7$ & kmeans & $0.05\pm0.06$ & $25.0\pm0.0$\tabularnewline
 & sparse-kmeans & $0.05\pm0.01$ & $15.0\pm0.0$\tabularnewline
 & Clustvarsel & $0.05\pm0.01$ & $25.0\pm2.0$\tabularnewline
 & Selvarclust & $0.05\pm0.01^{*}$ & $5.6\pm0.9^{*}$\tabularnewline
\cline{2-4} 
 & sparseFEM-1 & $0.04\pm0.01$ & $10.2\pm2.4$\tabularnewline
 & sparseFEM-2 & $0.05\pm0.01$ & $8.8\pm1.7$\tabularnewline
 & sparseFEM-3 & $0.04\pm0.01$ & $5.6\pm1.6$\tabularnewline
\hline 
\end{tabular}
\par\end{centering}

\caption{\label{tab:Simul_Sparse_comp}Clustering errors and numbers of non-zero
variables averaged over 20 simulations for several clustering methods
with $p=25$ and $q=5$. The results of Selvarclust are reported from~\cite{Celeux11}{\scriptsize .}}
\end{table}

This first experiment aims to compare on simulated data the performances
of the proposed sparseFEM algorithms (sparseFEM-1, sparseFEM-2, sparseFEM-3)
to several competitors: Clustvarsel of Raftery and Dean~\cite{Raftery06},
Selvarclust of Maugis \emph{et al.}~\cite{Maugis09a} and sparse-kmeans
of Witten and Tibshirani~\cite{Witten10}. For this experiment, we
replicated the simulation proposed in Section~3.3 of~\cite{Witten10}.
We simulated $K=3$ Gaussian components of $n$ observations in a
$25$-dimensional observation space whose components differ only on
$q=5$ features. The used parameters were $\mu_{kj}=\mu\times(\mathbf{1}_{k=1,j\leq q},-\mathbf{1}_{k=2,j\leq q}),\,\forall k\in\{1,2,3\}$
and $\forall j\in\{1,\dots,p\}$ for the mean components and $\sigma_{kj}^{2}=1$
for the variance terms. Moreover, four different situations are considered:
$n=30$ or $300$ and $\mu=0.6$ or $1.7$. Each simulation was replicated
$25$ times. 

Table~\ref{tab:Simul_Sparse_comp} presents the means and standard
deviations for both the clustering error and the number of non-zero
variables for kmeans, sparse-kmeans, Clustvarsel, Selvarclust and
the 3 procedures of sparseFEM. Note that the results of Selvarclust
corresponds to clustering errors and non-zero variable rates found
in~\cite{Celeux11}. Moreover, the reported results concerning the
3 sparse Fisher-EM algorithms were obtained with the DLM$_{[\alpha_{k}\beta]}$
model for a sparsity level corresponding to the highest BIC value
obtained at each trial. 

Two main remarks can be done on the results presented in Table~\ref{tab:Simul_Sparse_comp}.
First, by considering either the most difficult clustering cases ($n=30$
and $\mu=0.6$) or the easiest one ($n=300$ and $\mu=0.6$ or $1.7$),
all approaches present approximatively the same results in terms of
clustering error rate. The methods differ however in the number of
variables they retain to perform the clustering: Clustevarsel, sparseFEM-1,
sparseFEM-2 and sparseFEM-3 turn out to select significantly less
variables than sparse-kmeans and Clustvarsel. In particular, Clustevarsel
and the sparseFEM algorithms select a number of useful variables consistent
with the actual number of meaningful variables ($q=5$). Second, in
the situation where $n=30$ and $\mu=1.7$, Selvarclust and sparse-kmeans
present the lowest misclassification rate ($0.08$), even though the
clustering error of sparseFEM-1 and kmeans remains relatively low
($0.14$). However, as previously, only Clustevarsel and the sparseFEM
algorithms select a number of variables close to the right number
of discriminative features.

\subsection{Comparison on real data sets}

Real-world data sets are now used to compare the efficiency of the
sparseFEM algorithms to its competitors for both the clustering and
variable selection tasks. We considered 7 different benchmark data
sets coming mostly from the UCI machine learning repository. We selected
these data sets because they represent a wide range of situations
in term of number of observations $n$, number of variables $p$ and
number of groups $K$. These characteristics are given in the top
row of Table~\ref{Tab:sparse_UCI} and a detailed description of
these data sets can be found in~\cite{Bouveyron12a}.

For this experiment, we used the 3 sparseFEM algorithms and the 3
sparse methods introduced previously (sparse-kmeans, Clustvarsel and
Selvarclust). Since the evaluation of the clustering performance is
a complex and very discussed problem, we chose to evaluate the clustering
performance as the adequacy between the resulting partition of the
data and the known labels for these data. For each data set, the sparseFEM
algorithms were initialized with a common random partition drown from
a multinomial distribution with equal prior probabilities. For Clustvarsel,
Selvarclust and sparse-kmeans, the initialization was done with their
own deterministic procedure. Moreover, for each method, the number
$K$ of groups has been fixed to the actual one. For Clustvarsel,
Selvarclust and sparse-kmeans, the determination of the other free
parameters was done according to the tools provided by each approach.
For the sparseFEM algorithms, we used the penalized BIC criterion
to select the model and the level of sparsity. More precisely, we
first chose the model presenting the highest average BIC value on
20 replications. Then, given the selected model, we selected the level
of sparsity associated with the highest BIC value.

\begin{sidewaystable}
\noindent \begin{centering}
\begin{tabular}{|l|c|c|c|c|c|c|c|}
\hline 
 & \multicolumn{1}{c|}{iris } & \multicolumn{1}{c|}{wine } & chiro & zoo  & glass  & satimage  & usps358\tabularnewline
 & \textit{(p=4,K=3)} & \textit{(p=13,K=3)} & \textit{(p=17,K=3)} & \textit{(p=16,K=7)} & \textit{(p=9,K=7)} & \textit{(p=36,K=6)} & \textit{(p=256,K=3)}\tabularnewline
Approaches & \textit{(n=150)} & \textit{(n=178)} & \textit{(n=178)} & \textit{(n=101)} & \textit{(n=214)} & \textit{(n=4435)} & \textit{(n=1726)}\tabularnewline
\hline 
\hline 
sparseFEM-1 & \cellcolor[gray]{0.9}96.5$\pm0.3$ & 97.8$\pm0.2$ & 84.2$\pm11$ & 71.4$\pm8.5$ & 50.2$\pm1.9$ & 69.6$\pm0.1$ & \cellcolor[gray]{0.9}84.7$\pm3.2$\tabularnewline
 & \textit{(2.0$\pm0.0$)} & \textit{(2.0$\pm0.0$)} & \textit{(2.}3\textit{$\pm0.5$)} & \textit{(13$\pm2.5$)} & \textit{(6.0$\pm1.0$)} & \textit{(36$\pm0.0$)} & \textit{(5.5$\pm0.7$)}\tabularnewline
sparseFEM-2 & 89.9$\pm0.4$ & \cellcolor[gray]{0.9}98.3$\pm0.0$ & 84.8$\pm12$ & 70.1$\pm12.2$ & 48.4$\pm3.0$ & 67.5$\pm1.6$ & 82.8$\pm9.1$\tabularnewline
 & \textit{(4.0$\pm0.0$)} & \textit{(4.0$\pm0.0$)} & \textit{(2.0$\pm0.6$)} & \textit{(14$\pm3.6$)} & \textit{(6.6$\pm0.7$)} & \textit{(36$\pm.0.0$)} & \textit{(15.5$\pm16$)}\tabularnewline
sparseFEM-3 & \cellcolor[gray]{0.9}96.5$\pm0.3$ & 97.8$\pm0.0$ & 82.9$\pm12$ & 72.0$\pm4.3$ & 48.2$\pm2.7$ & \cellcolor[gray]{0.9}71.8$\pm2.3$ & 79.1$\pm7.4$\tabularnewline
 & \textit{(2.0$\pm0.3$)} & \textit{(2.0$\pm0.0$)} & \textit{(2.0$\pm0.0$)} & \textit{(10$\pm2.8$)} & \textit{(7.0}$\pm0.0$\textit{)} & \textit{(36$\pm0.0$)} & \textit{(6.0$\pm1.3$)}\tabularnewline
\hline 
sparse-kmeans & 90.7 & 94.9 & \cellcolor[gray]{0.9}95.3 & 79.2 & \cellcolor[gray]{0.9}52.3 & 71.4 & 74.7\tabularnewline
 & \textit{(4.0)} & \textit{(13.0)} & \textit{(17.0)} & \textit{(16.0)} & \textit{(6.0)} & \textit{(36.0)} & \textit{(213)}\tabularnewline
\hline 
Clustvarsel & 96.0 & 92.7 & 71.1 & 75.2 & 48.6 & 58.7 & 48.3\tabularnewline
 & \textit{(3.0)} & \textit{(5.0)} & \textit{(6.0)} & \textit{(3.0)} & \textit{(3.0)} & \textit{(19.0)} & \textit{(6.0)}\tabularnewline
Selvarclust & 96.0 & 94.4 & 92.6 & \cellcolor[gray]{0.9}92.1 & 43.0 & 56.4 & 36.7\tabularnewline
 & \textit{(3.0)} & \textit{(5.0)} & \textit{(8.0)} & \textit{(5.0)} & \textit{(6.0)} & \textit{(22.0)} & \textit{(5.0)}\tabularnewline
\hline 
\end{tabular}
\par\end{centering}

\caption{\label{Tab:sparse_UCI}Clustering accuracies and their standard deviations
(in percentage){\small{} }on 7 UCI datasets (iris, wine, chironomus,
zoo, glass, satimage, usps358) averaged on 20 trials. The average
number of nonzero variables is reported in brackets. No standard deviation
is reported for Clustvarsel/Selvarclust and sparse-kmeans since their
initialization procedure is deterministic and always provides the
same initial partition. }
\end{sidewaystable}

Table~\ref{Tab:sparse_UCI} presents the average clustering accuracies
and the associated standard deviations obtained for the 6 approaches.
The average number of non-zero variables is also reported within brackets
in the table. The results associated to the sparseFEM algorithms have
been obtained by averaging over $20$ trials with random initializations.
The lack of standard deviations for Clustvarsel, Selvarclust and sparse-kmeans
is due to the deterministic initializations they use. It first appears
that the three sparse versions of the Fisher-EM algorithm perform
rather similarly both in term of clustering and variable selection.
It also appears clearly that the sparseFEM algorithms are competitive
to existing methods regarding both the clustering performances and
the selection of variables. Indeed, the sparseFEM algorithms obtain
the best clustering accuracies on 4 of the 7 data sets whereas sparse-kmeans
and Selvarclust obtain the best clustering accuracies on respectively
2 and 1 data sets. The sparseFEM algorithms differ also from sparse-kmeans
regarding the number of variables retained to perform the clustering.
Indeed, sparse-kmeans turns out to frequently select a large number
of variables whereas sparseFEM is usually rather sparse in the number
of selected variables. Finally, Clustvarsel and Selvarclust turn out
to select most of the time few variables, particularly in high-dimensional
spaces, which seems to obstruct their clustering performance. To summarize,
this experiment has shown that the sparseFEM algorithms seem to be
good compromises between sparse-kmeans and Clustvarsel /Selvarclust
in term of variable selection and, certainly thanks to this characteristic,
they also provide good clustering results.

\subsection{Comparison on the usps358 data set}

We focus now on the usps358 dataset to stress the role of variable
selection in the interpretation of clustering results. The original
dataset is made of $7\,291$ images divided in 10 classes corresponding
to the digits from 0 to 9. Each digit is a $16\times16$ gray level
image represented as a 256-dimensional vector. For this experiment,
we extracted a subset of the data ($n=1\,756$) corresponding to the
digits 3, 5 and 8 which are the three most difficult digits to discriminate.
This smaller dataset is hereafter called usps358. Figure~\ref{fig:USPS_benchmark}
depicts the group mean images obtained from the true labels in the
usps358 dataset. For this experiment, we used the three sparseFEM
algorithms with the model and the level of sparsity selected in the
previous experiment for this data set. For Clustvarsel, Selvarclust
and sparse-kmeans, the level of sparsity was again selected with their
own selection procedure.

\begin{figure}[p]
\begin{centering}
\subfloat[]{\includegraphics[bb=55bp 60bp 480bp 475bp,clip,width=0.2\columnwidth]{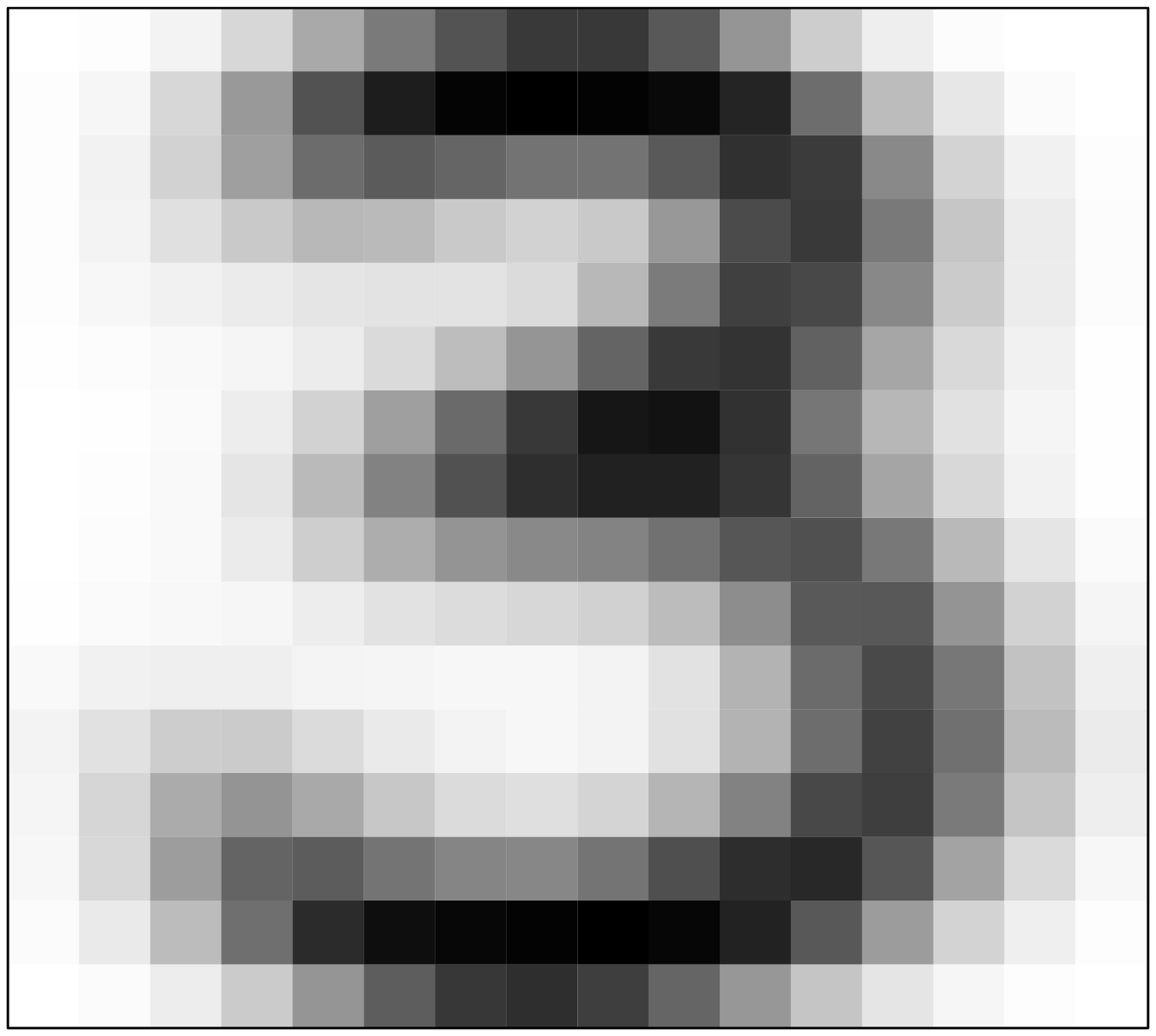}

}\hspace{1.5cm}\subfloat[]{\includegraphics[bb=55bp 60bp 480bp 475bp,clip,width=0.2\columnwidth]{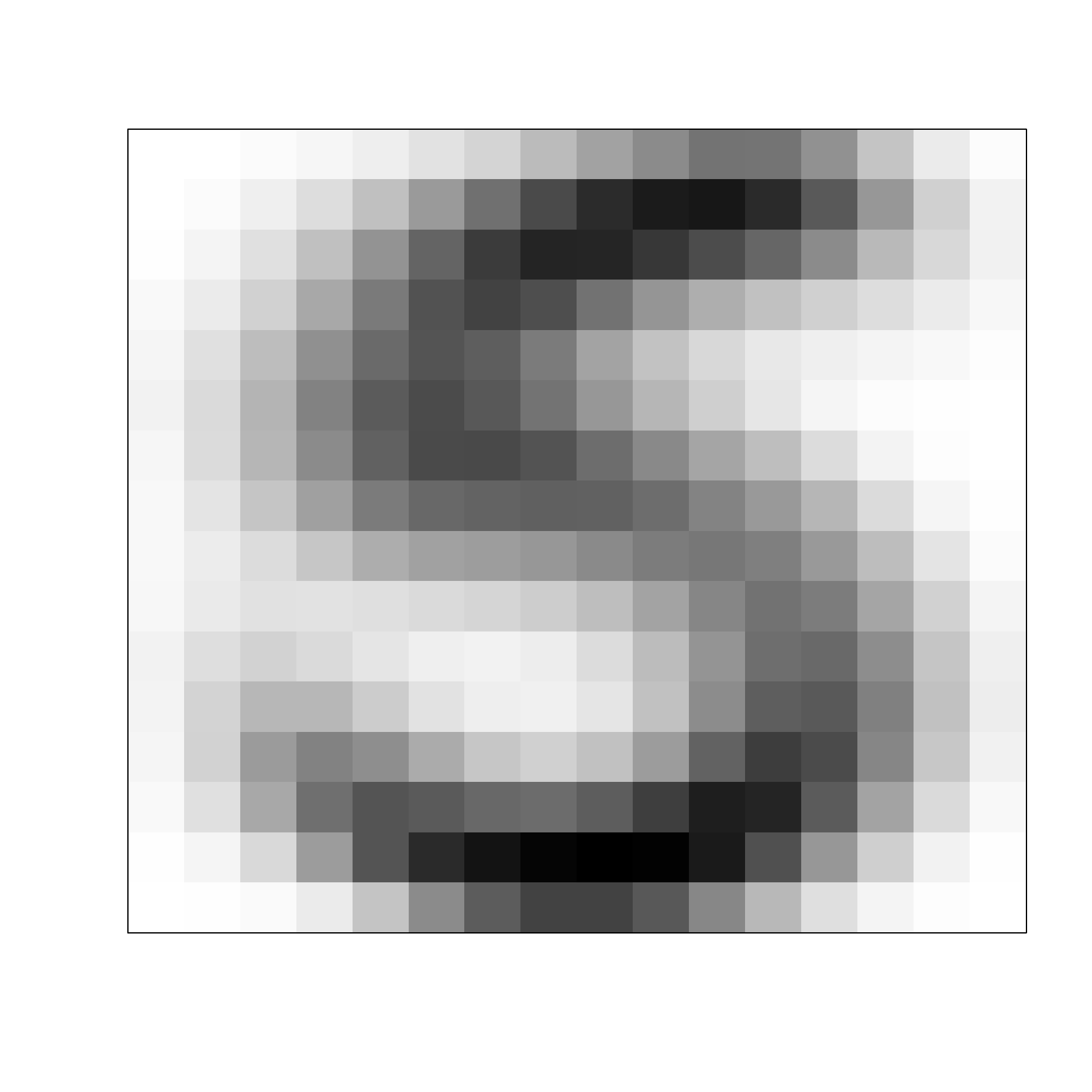}

}\hspace{1.5cm}\subfloat[]{\includegraphics[bb=55bp 60bp 480bp 475bp,clip,width=0.2\columnwidth]{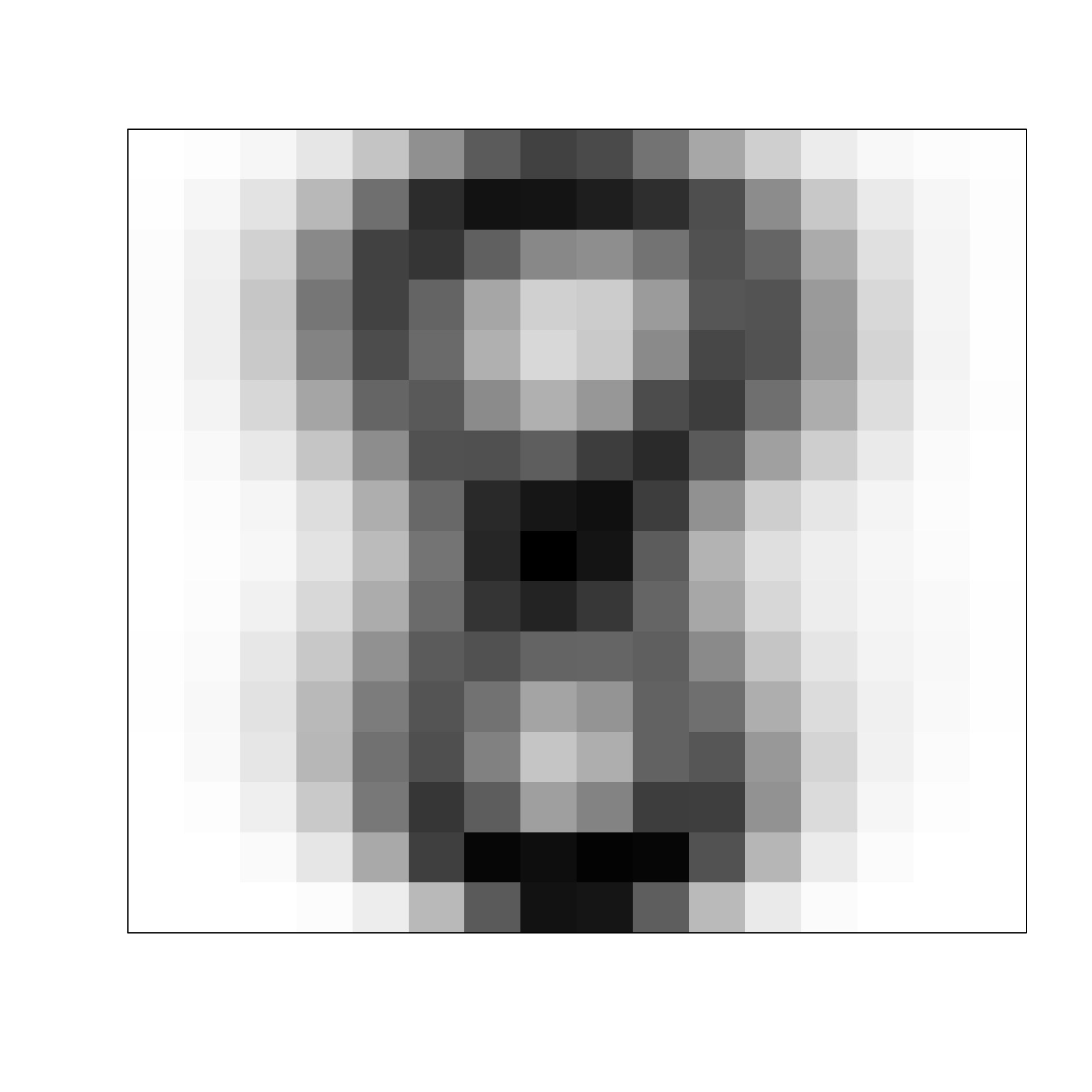}

}
\par\end{centering}

\caption{\label{fig:USPS_benchmark}Group means obtained from the true labels
in the USP358 datasets.}

\begin{centering}
\subfloat[Sparse-kmeans]{\includegraphics[bb=55bp 60bp 480bp 475bp,clip,width=0.2\columnwidth]{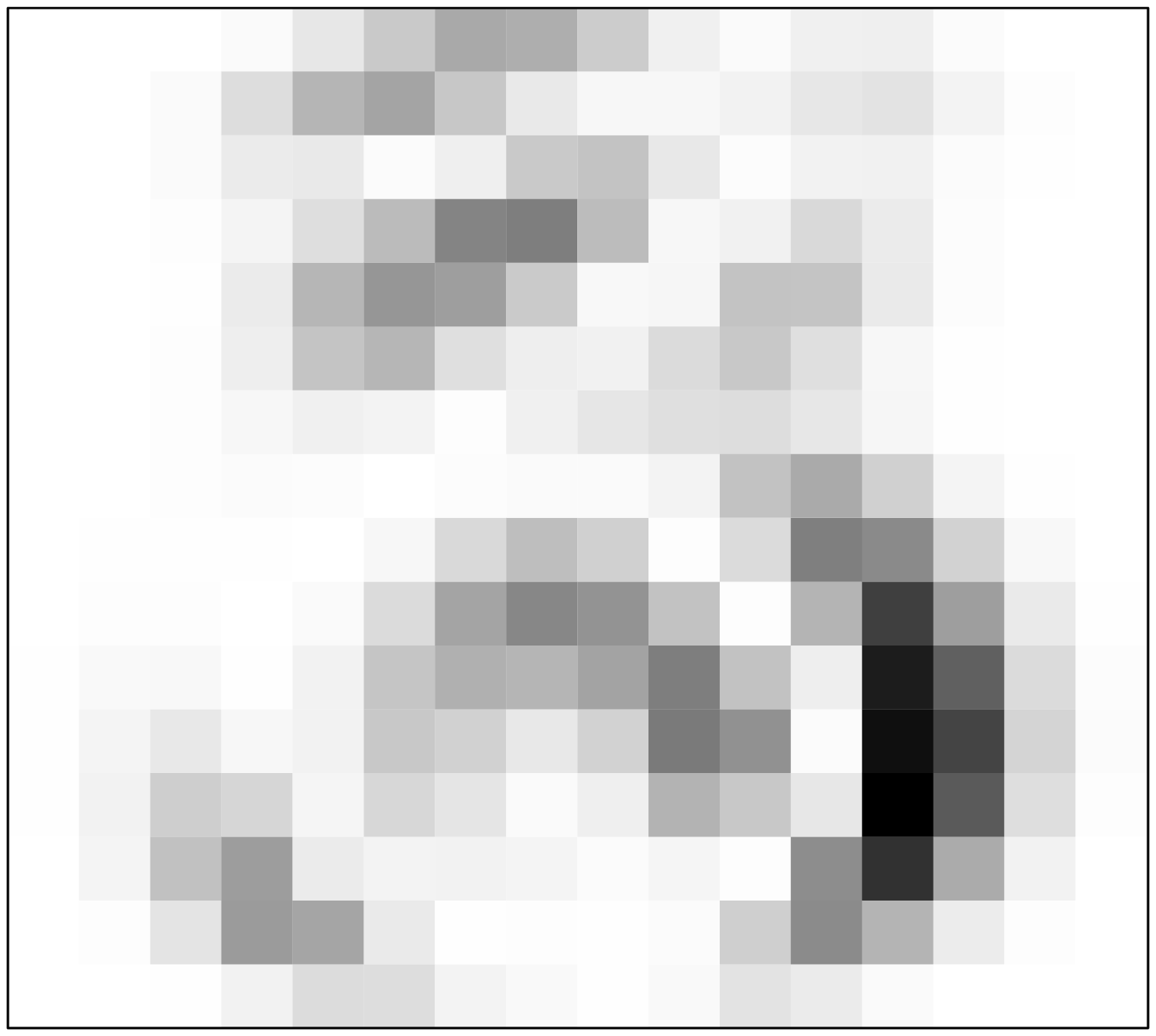}

}\hspace{1.5cm}\subfloat[Clustvarsel]{\includegraphics[bb=55bp 60bp 480bp 475bp,clip,width=0.2\columnwidth]{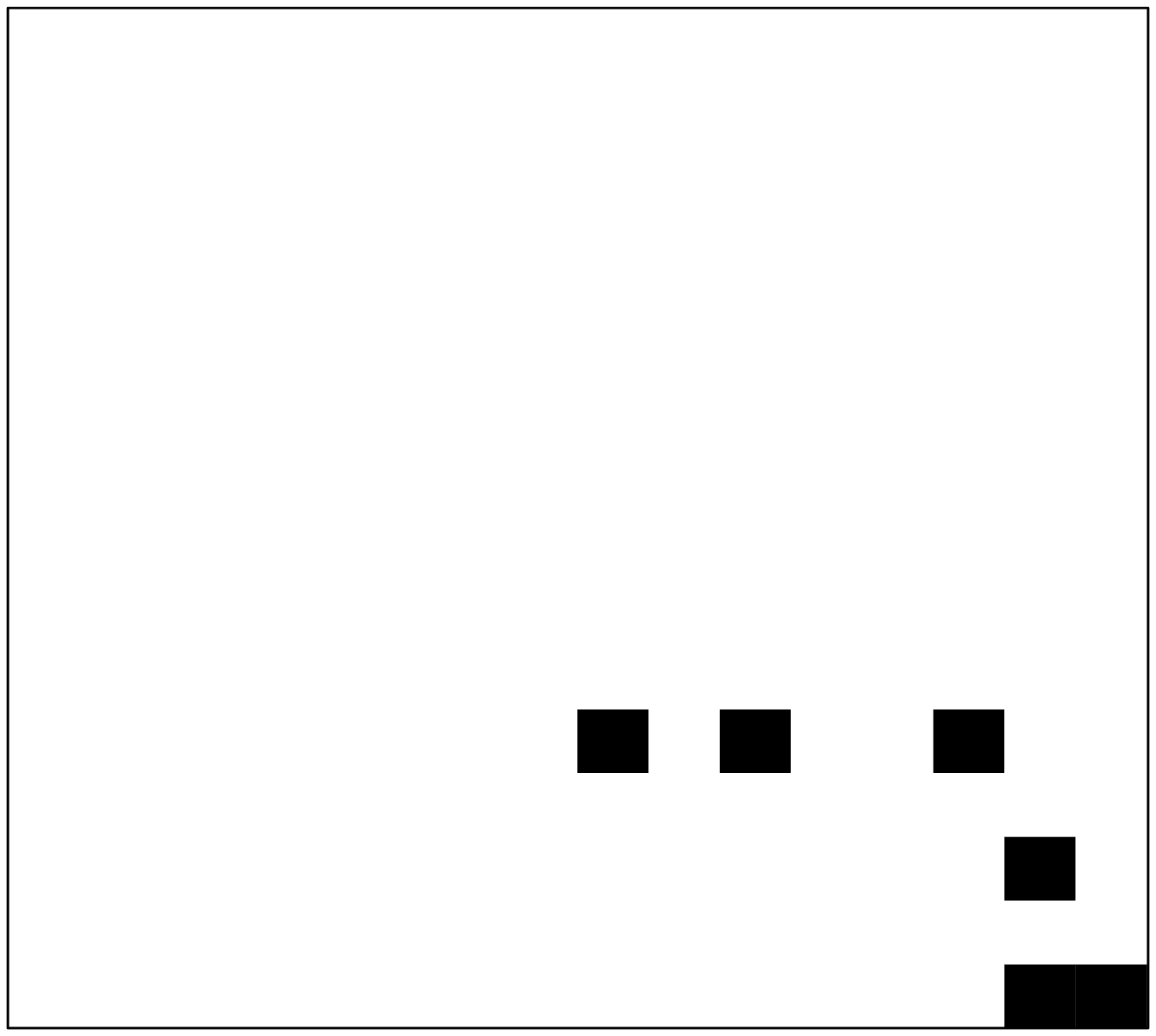}

}\hspace{1.5cm}\subfloat[Selvarclust]{\includegraphics[bb=55bp 60bp 480bp 475bp,clip,width=0.2\columnwidth]{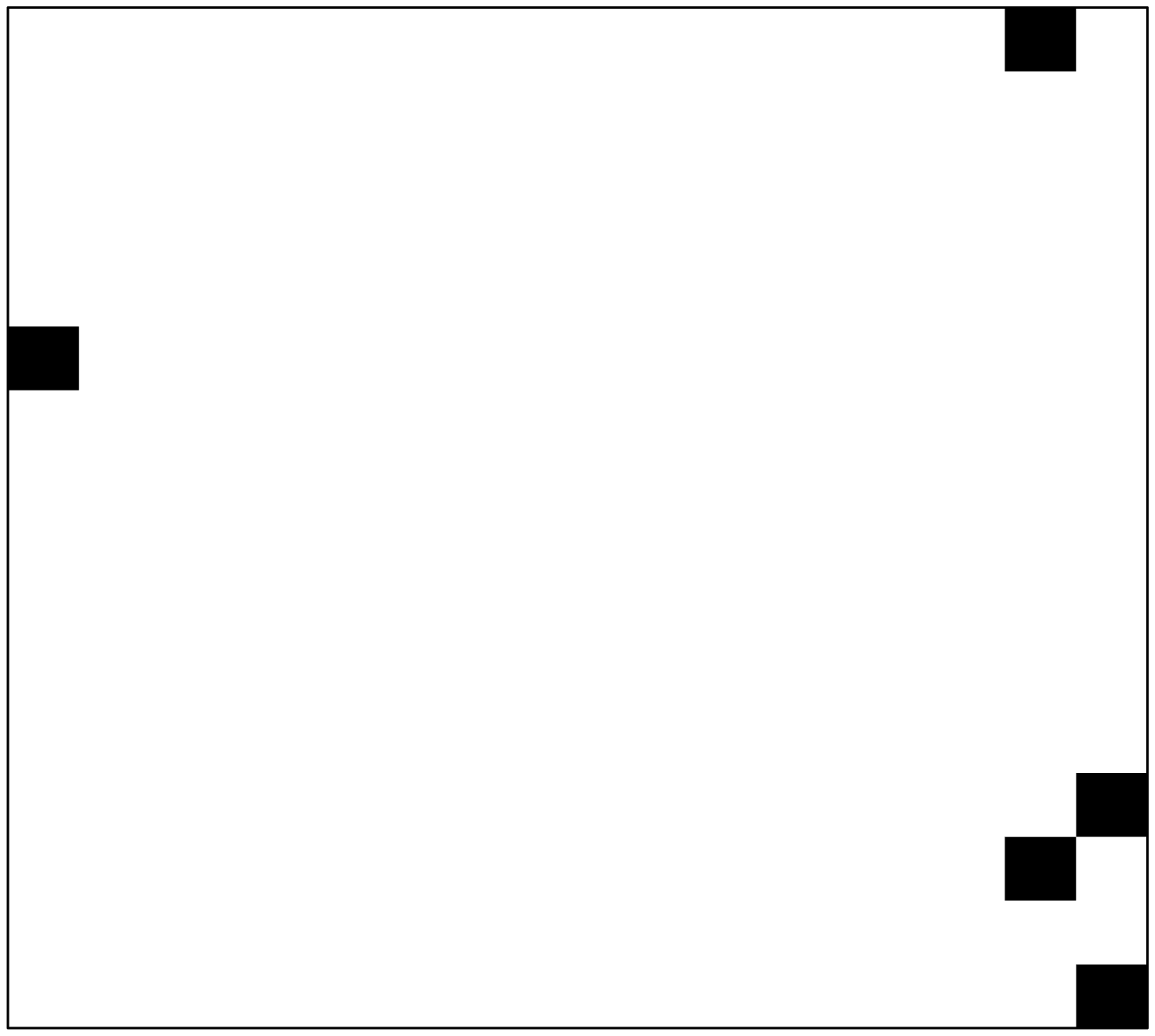}

}
\par\end{centering}

\caption{\label{fig:USPS_existingApp}Variable selection obtained from (a)
the sparse-kmeans algorithm, (b) the Clustvarsel approach and (c)
the Selvarclust approach.}

\begin{centering}
\subfloat[sparseFEM-1]{\includegraphics[bb=55bp 60bp 480bp 475bp,clip,width=0.2\columnwidth]{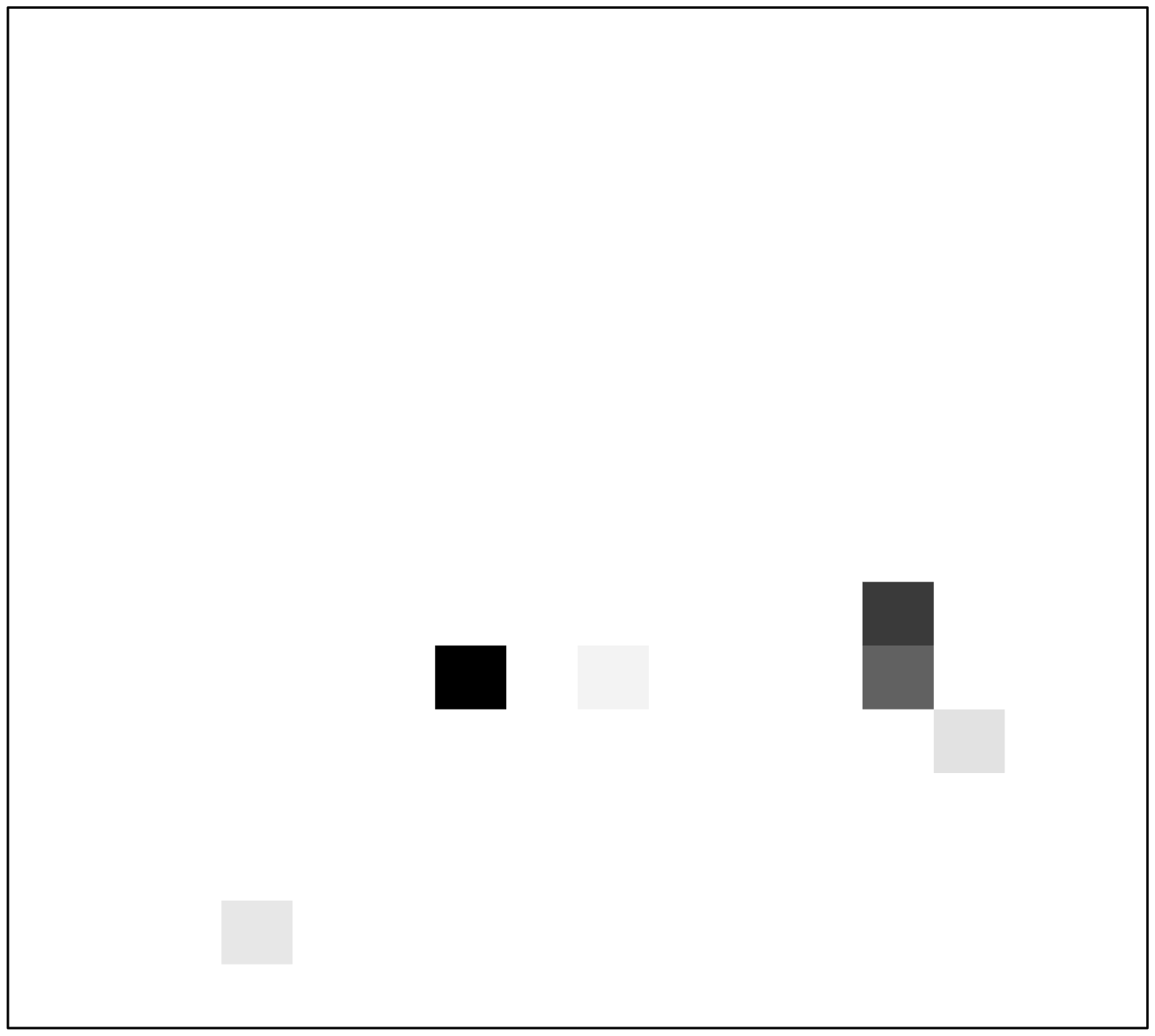}

}\hspace{1.5cm}\subfloat[sparseFEM-2]{\includegraphics[bb=55bp 60bp 480bp 475bp,clip,width=0.2\columnwidth]{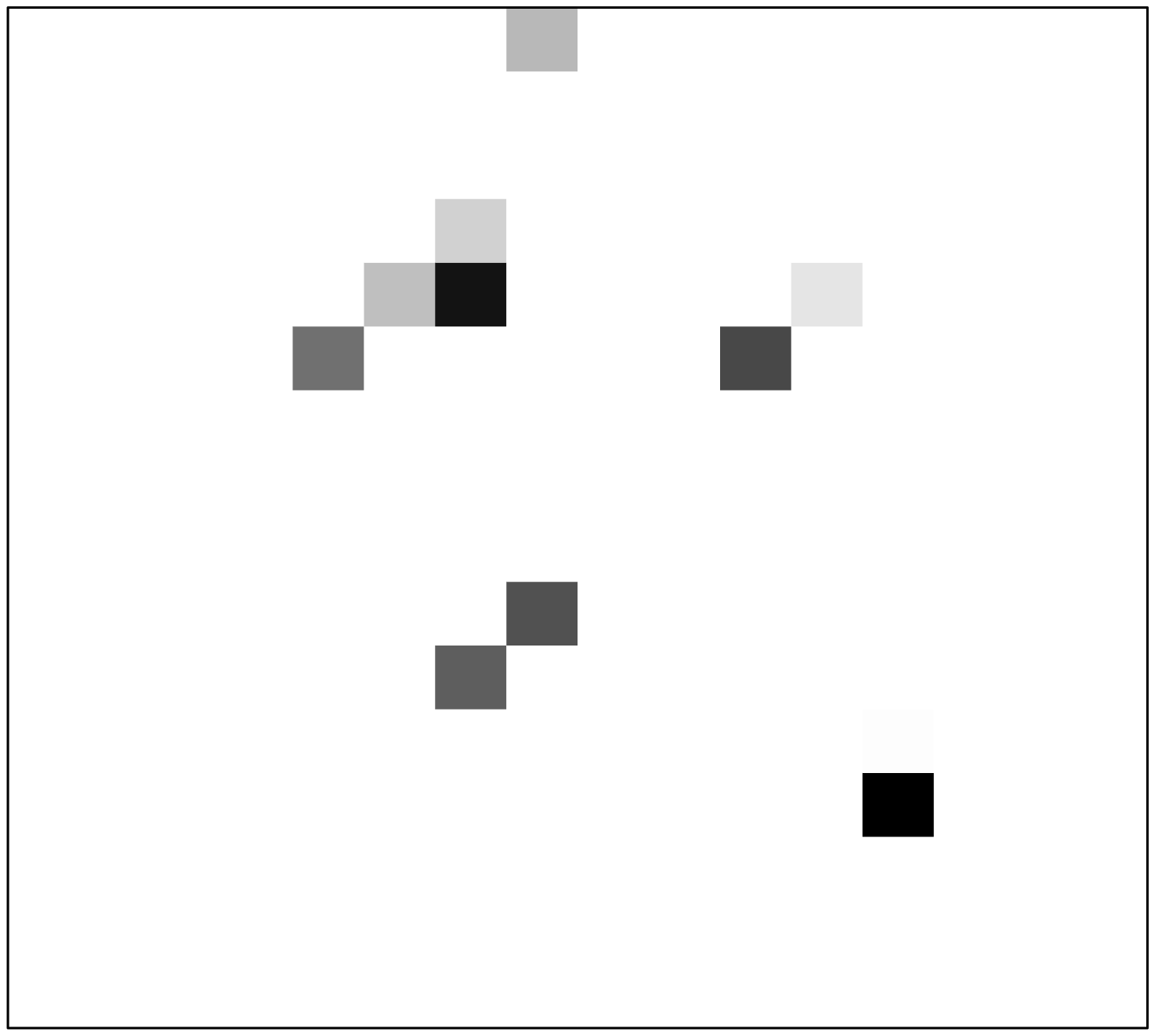}

}\hspace{1.5cm}\subfloat[sparseFEM-3]{\includegraphics[bb=55bp 60bp 480bp 475bp,clip,width=0.2\columnwidth]{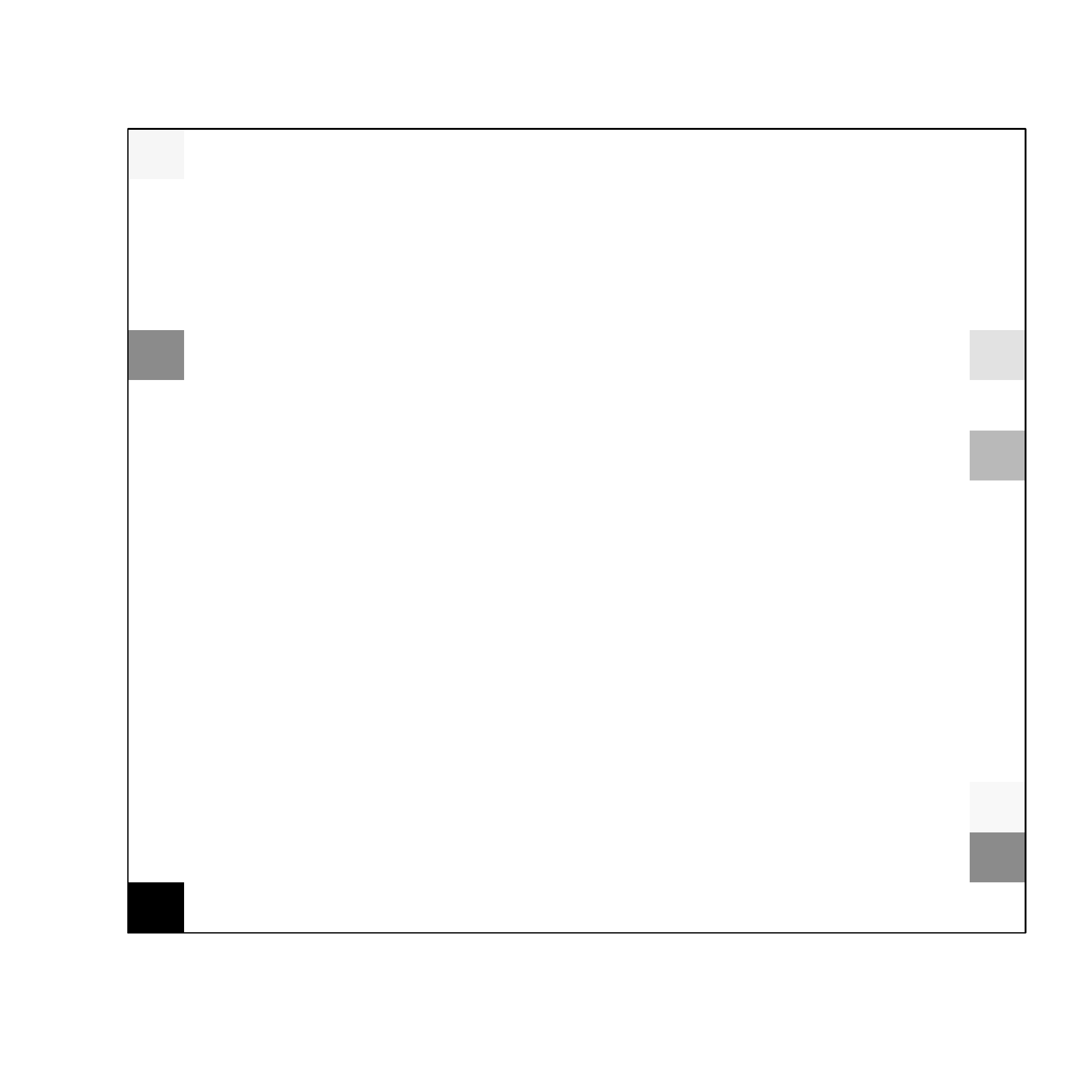}

}
\par\end{centering}

\caption{\label{fig:USPS_SFEM}Variable selection obtained from (a) the sparseFEM-1,
(b) the sparseFEM-2 and (c) the sparseFEM-3 procedures with sparsity
levels selected by the penalized~BIC.}
\end{figure}

Figures~\ref{fig:USPS_existingApp} illustrates, as images, the features
selected respectively by sparse-kmeans~(Figure~\ref{fig:USPS_existingApp}.a),
Clustvarsel (Figure~\ref{fig:USPS_existingApp}.b) and Selvarclust
(Figure~\ref{fig:USPS_existingApp}.c). In Figure~\ref{fig:USPS_existingApp}.a,
the weight assigned by sparse-kmeans to each feature is represented
by gray levels: lighter is the pixel, weaker is the absolute value
of the weight of the associated feature. For Clustvarsel and Selvarclust,
only the selected variables are depicted and are associated to black
pixels as it is illustrated in Figures~\ref{fig:USPS_existingApp}.b
and \ref{fig:USPS_existingApp}.c respectively. These representations
are associated to the following clustering accuracies $74.7\%$, $48.3\%$
and $36.7\%$ for sparse-kmeans, Clustvarsel and Selvarclust respectively.
For the 3 sparseFEM algorithms, we superimposed in a same figure the
absolute values of the loadings of the two discriminative axes fitted
by the sparseFEM-1, sparseFEM-2 and sparseFEM-3 procedures. The associated
clustering accuracies are respectively $84.7\%$, $82.8\%$ and $79.1\%$.

First of all, it appears that Clustvarsel and Selvarclust select significantly
fewer variables than both sparse-kmeans or the sparseFEM procedures.
Furthermore, most of the selected variables by Clustvarsel and Selvarclust
turn out to be irrelevant to discriminate the digit 3 from the digits
5 and 8. For instance, in Figures~\ref{fig:USPS_existingApp}.b and
\ref{fig:USPS_existingApp}.c, we can observe that the black pixels
located in right bottom corner, do not correspond to any discriminative
variable. This certainly explain the poor clustering performances
($48.3\%$ for Clustvarsel and $36.7\%$ for Selvarclust) observed
on this data set for these methods. On the contrary, sparse-kmeans
turns out to perform well in term of clustering performance ($74.7\%$
of clustering accuracy). Nevertheless, the number of selected variables
remains higher (213 selected variables amongst 256 original ones)
than we would expect to ease the interpretation of results. Finally,
sparseFEM-1 and sparseFEM-2 seem to answer quite well to both the
clustering task and the task of feature selection. Indeed, on the
one hand, the subset of selected pixels remains small for both algorithms:
6 and 15 pixels are selected amongst 256 for sparseFEM-1 and sparseFEM-2
respectively. Furthermore, the selected pixels appear to be relevant
to discriminate the classes associated with the three digits. For
instance, the darker pixel on the bottom right corner of Figure~\ref{fig:USPS_SFEM}.b
discriminates the digit 8 from the digits 3 and 5. On the other hand,
and certainly due to this relevant selection of variables, both algorithms
perform particularly well on this high-dimensional data set ($84.7\%$
for sparseFEM-1 and $82.8\%$ for sparseFEM-2). However, on this data
set, the sparseFEM-3 procedure shows a disappointing behavior regarding
the variable selection even though its clustering performance remains
satisfying. The fact that sparseFEM-3 succeeds in clustering the data
set even with a bad selection of variables is certainly due to the
nature of the DLM model which models also the non discriminative information
through the parameter $\beta_{k}$.

Table~\ref{CPU_usps358} presents the computing time of the studied
clustering methods (for a given model and with $\lambda$ and $K$
fixed) for clustering the usps358 data set. As we can remark, our
procedures are much faster than the sparse-kmeans, Clustvarsel and
Selvarclust algorithms. Consequently, the sparseFEM algorithms appear
once again to be good compromises, in practice, to cluster high-dimensional
data and select a set of discriminative variables in a reasonable
time.

\begin{table}
\begin{centering}
\begin{tabular}{lc|lc}
Approaches: & Procedure time (sec) & Approaches: & Procedure time (sec)\tabularnewline
\hline 
\hline 
sparseFEM-1 & $729.04$ & sparse-kmeans & $1\,567.75$\tabularnewline
sparseFEM-2 & \cellcolor[gray]{0.9}$387.12$ & Clustvarsel & $2\,957.70$\tabularnewline
sparseFEM-3 & $409.61$ & Selvarclust & $9\,257.10$\tabularnewline
\hline 
\end{tabular}
\par\end{centering}

\caption{\label{CPU_usps358}Computing times for the 3 versions of the sparseFEM
algorithm, sparse-kmeans, Clustvarsel and Selvarclust on the USPS358
data (for a given model and with $\lambda$ and $K$ fixed).}
\end{table}

\section{Application to the segmentation of hyperspectral images}

\begin{figure}
\begin{centering}
\includegraphics[bb=55bp 70bp 705bp 495bp,clip,width=0.6\columnwidth,height=0.5\columnwidth]{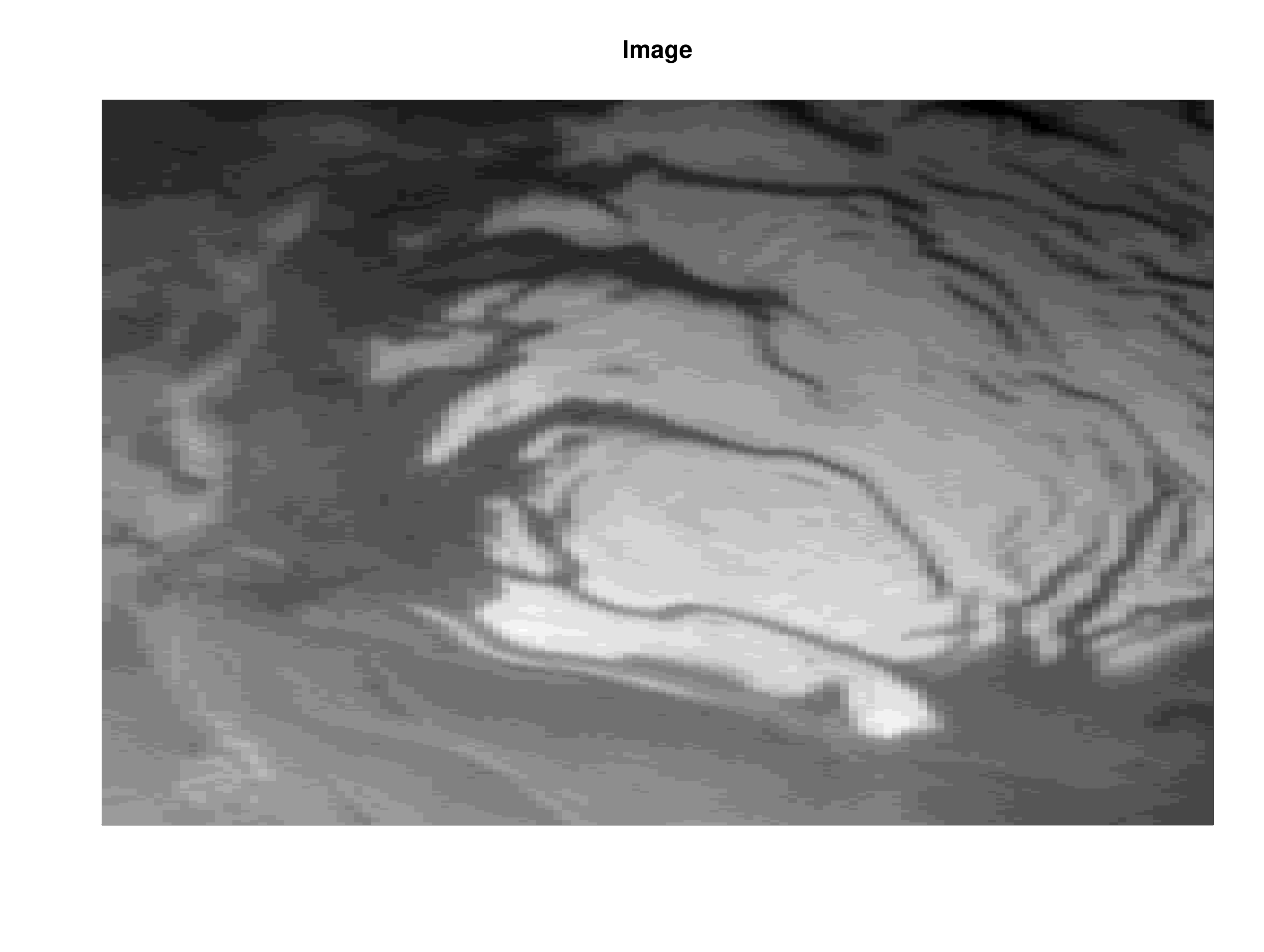}
\par\end{centering}

\caption{\label{fig:Mars-1}Image of the studied zone of the Martian surface.}
\end{figure}

\begin{figure}
\begin{centering}
\includegraphics[bb=30bp 45bp 765bp 380bp,clip,width=1\columnwidth]{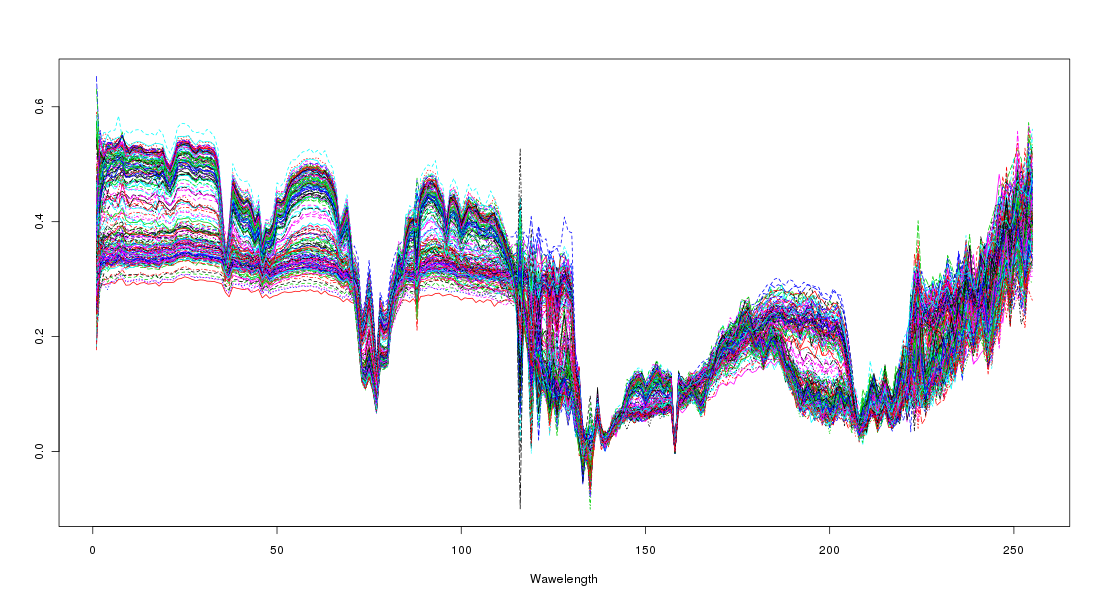}
\par\end{centering}

\caption{\label{fig:Mars-2}Some of the 38~400 measured spectra described
on 256 wavelengths (see text for details).}
\end{figure}

Here, we propose to use sparseFEM to segment hyperspectral images
of the Martian surface. Visible and near infrared imaging spectroscopy
is a key remote sensing technique to study the system of the planets.
Imaging spectrometers, which are inboard of an increasing number of
satellites, provide high-dimensional hyperspectral images. In March
2004, the OMEGA instrument (Mars Express, ESA)~\cite{Bibring2004}
has collected 310 Gbytes of raw images. The OMEGA imaging spectrometer
has mapped the Martian surface with a spatial resolution varying between
300 to 3000 meters depending on the spacecraft altitude. It acquired
for each resolved pixel the spectrum from 0.36 to 5.2 µm in 256 contiguous
spectral channels. OMEGA is designed to characterize the composition
of surface materials, discriminating between various classes of silicates,
hydrated minerals, oxides and carbonates, organic frosts and ices.
For this experiment, a $300\times128$ image of the Martian surface
is considered and a 256-dimensional spectral observation is therefore
associated to each of the 38~400 pixels. Figure~\ref{fig:Mars-1}
presents an image of the studied zone and Figure\textbf{~}\ref{fig:Mars-2}
shows some of the 38~400 measured spectra. According to the experts,
there are $K=5$ mineralogical classes to identify.

The sparseFEM-1 algorithm was applied to this dataset using the model
$\mathrm{DLM}_{[\alpha_{kj}\beta]}$ and a sparsity ratio equals to
$0.1$ (it refers to the ratio of the $\ell_{1}$ norm of the coefficient
vector relative to the norm at the full least square solution). The
sparseFEM algorithm was initialized with the results of the Fisher-EM
algorithm and the whole segmentation process took 18 hours on a 2.6
Ghz computer. Figure~\ref{fig:Mars-3} presents, on the right panel,
the segmentation into $5$ mineralogical classes of the studied zone
with the sparseFEM algorithm. In comparison, the left panel of Figure~\ref{fig:Mars-3}
shows the segmentation obtained by experts of the domain using a physical
model. It first appears that the two segmentations agree globally
on the mineralogical nature of the surface of the studied zone (60.30\%
of agreement). We recall that both segmentations do not exploit the
spatial information. When looking at the top-right quarter of the
image, we can notice that sparseFEM seems to provide a finer segmentation
than the segmentation based on the physical model. Indeed, sparseFEM
segments better than the physical model the fine {}``rivers'' which
can be seen on Figure~\ref{fig:Mars-1}. 

\begin{figure}
\begin{centering}
\begin{tabular}{cc}
\includegraphics[bb=55bp 70bp 705bp 495bp,clip,width=0.48\columnwidth,height=0.42\columnwidth]{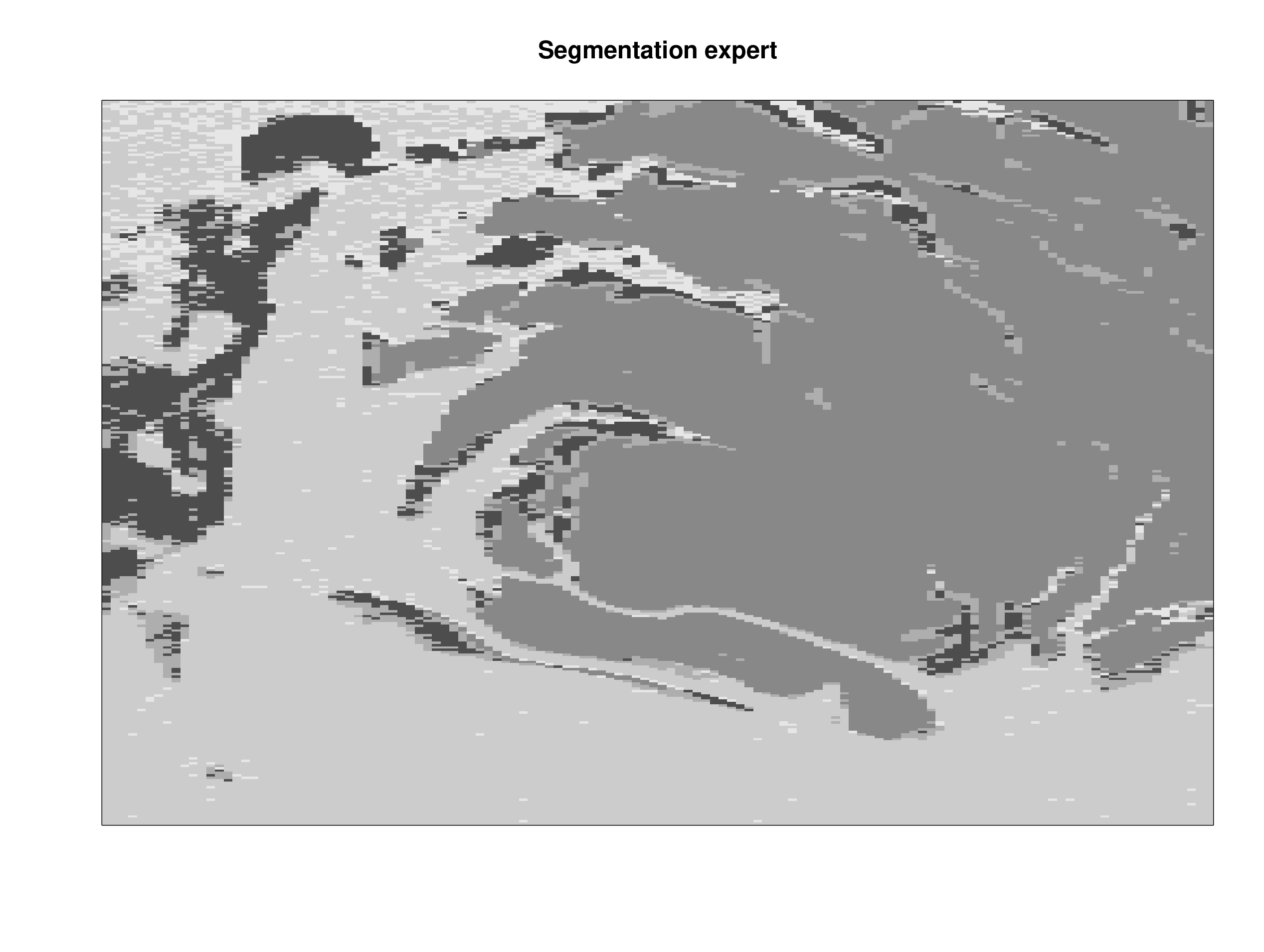} & \includegraphics[bb=55bp 70bp 705bp 495bp,clip,width=0.48\columnwidth,height=0.42\columnwidth]{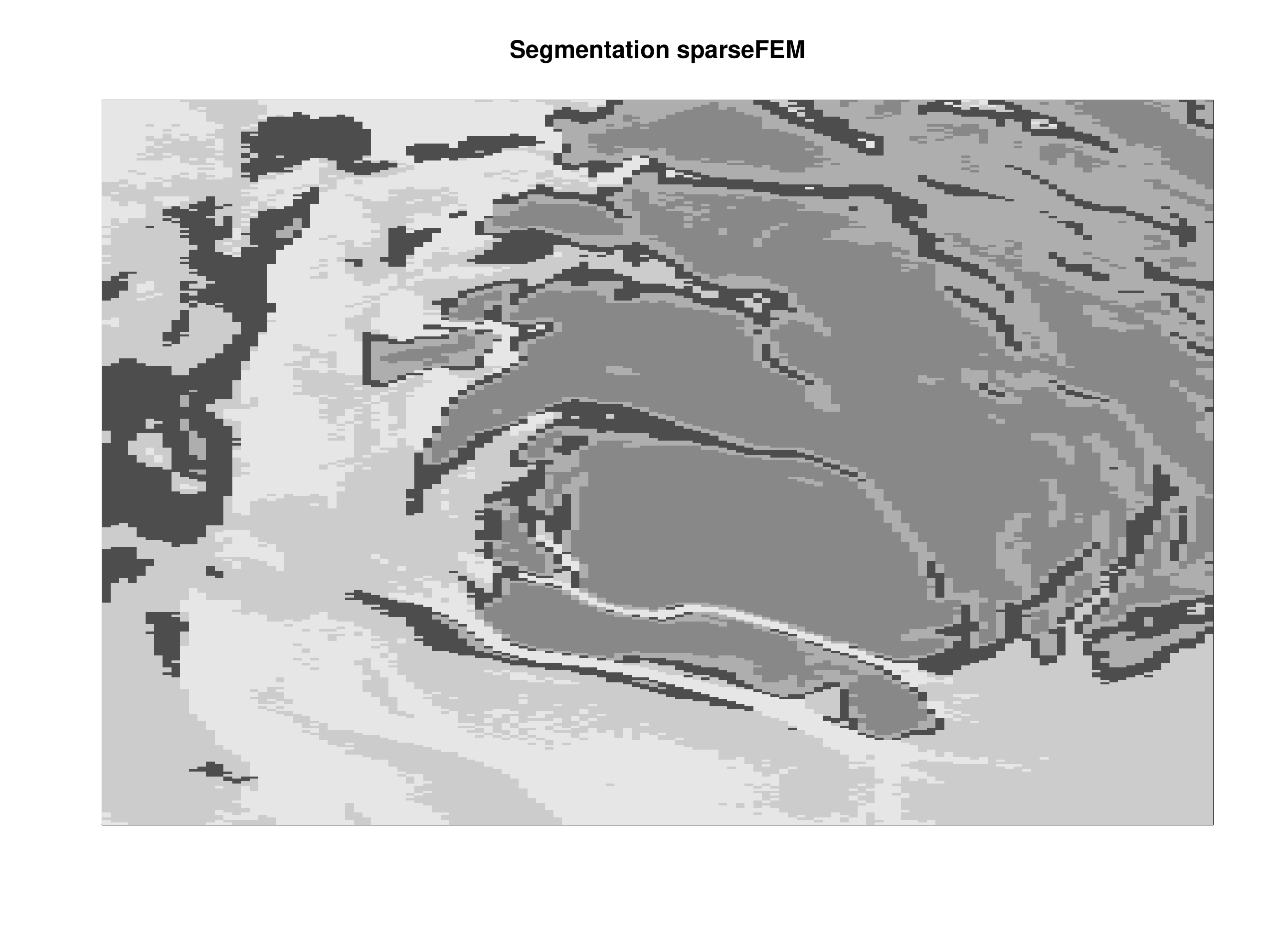}\tabularnewline
Expert segmentation & SparseFEM segmentation\tabularnewline
\end{tabular}
\par\end{centering}

\caption{\label{fig:Mars-3}Segmentation of the hyperspectral image of the
Martian surface using a physical model build by experts (left) and
sparseFEM (right).}
\end{figure}

\begin{figure}
\begin{centering}
\includegraphics[bb=30bp 45bp 910bp 380bp,clip,width=1\columnwidth]{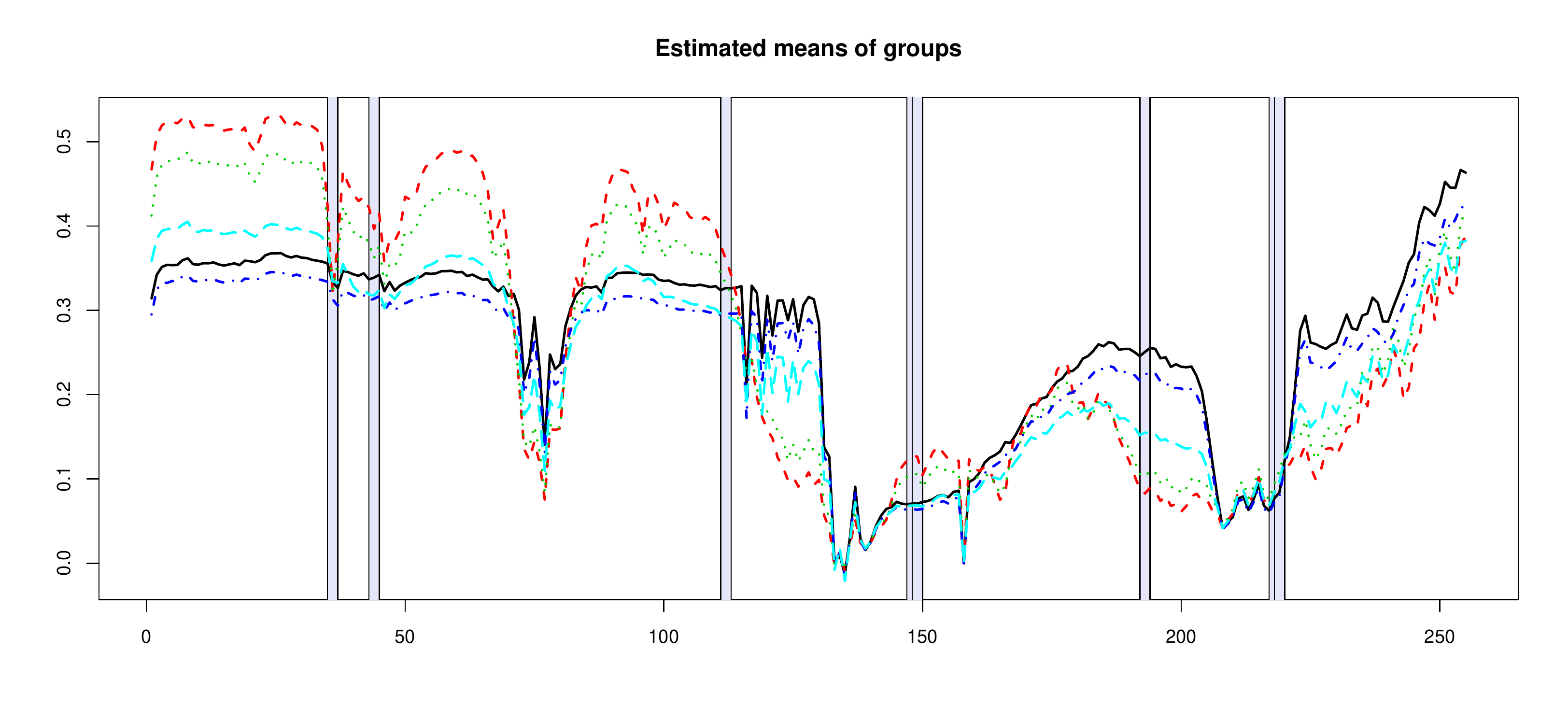}
\par\end{centering}

\caption{\label{fig:Mars-4}Mean spectra of the 5 groups formed by sparseFEM
and selection of the discriminative wavelengths (indicated by gray
rectangles).}
\end{figure}

Finally, Figure~\ref{fig:Mars-4} shows the mean spectra of the 5
groups formed by sparseFEM and the selection of the discriminative
wavelengths. SparseFEM selected 8 original variables (wavelengths)
as discriminative variables, \emph{i.e.} the rows associated to these
variables were non-zero in the loading matrix $U$. Looking closely
at the selection, we indeed notice that the first selected variable
(from left to right) discriminates the blue group from the others.
The second selected variable discriminates the red and green groups
from the black, blue and light blue groups whereas the third selected
variable allows to discriminate the red, green and black groups from
the blue and light blue groups. Similarly, the fourth and fifth selected
variables discriminate the red and green groups from the black, blue
and light blue groups whereas the sixth, seventh and eighth selected
variable allows to discriminate the red, green and light blue groups
from the blue and black groups.

A possible interest of such a selection could be the measurement of
only a tens of wavelengths for future acquisitions instead of the
256 current ones for a result expected to be similar. This could in
particular reduce the acquisition time for each pixel from a few tens
of seconds to less than one second.

\section{Conclusion}

This article has focused on variable selection for clustering with
the Fisher-EM algorithm which has been recently proposed in~\cite{Bouveyron12a}.
The aim of this work was to introduce sparsity in the Fisher-EM algorithm
and thus select the discriminative variables among the set of original
variables. We have proposed three different procedures based on a
$\ell_{1}$-penalty term. Experiments on simulations and real data
sets have shown that the three sparse versions of the Fisher-EM algorithm
are highly competitive with existing approaches of the literature.
In particular, the sparseFEM procedures present several assets regarding
existing approaches. On the one hand, they tend to select an intermediate
number of discriminative variables whereas existing approaches tend
to select either too few (Clustvarsel and Selvarclust) or too much
variables (sparse-kmeans). On the other hand, the sparseFEM procedures
perform both the clustering and the variable selection in a reasonable
time comparing to existing approaches in the case of high-dimensional
data. The sparseFEM algorithms have been also applied with success
to the segmentation of hyperspectral images of the planet Mars and
relevant parts of the spectra which well discriminate the groups have
been identified.

Among the possible extensions of this work, it may be first interesting
to use different $\ell_{1}$-penalty values according to the relevance
of each discriminative axis estimated in the Fisher-EM algorithm.
Such an approach could identify different levels of relevancy among
the original variables. Second, we used in this work a penalized BIC
criterion to select the sparsity level by evaluating the model complexity
in regard to the non-zero values as proposed by~\cite{Pan07}. Although
Zou \emph{et al}.~\cite{Zou07} showed that the number of non-zero
coefficients is an unbiased estimate of the degrees of freedom and
is asymptotically consistent in the case of penalized regression problem,
this result has no theoretical justification in the penalized GMM
context. It would be therefore interesting to obtain theoretical guarantees
of such a result in our context. Finally, since the ICL criterion~\cite{Biernacki01}
is also used to select the number of components, it would be a natural
extension to consider a penalized ICL for selecting the sparsity level
in the sparseFEM algorithms.

\section*{Acknowledgments}

The authors would like to thank Cathy Maugis for providing the results
of Selvarclust on the zoo, glass, satimage and usps358 data sets.

\bibliographystyle{plain}
\bibliography{bibliFEM}

\begin{thebibliography}{10}

\bibitem{Baek09}
J.~Baek and G.~McLachlan.
\newblock Mixtures of factor analyzers with common factor loadings:
  applications to the clustering and visualisation of high-dimensional data.
\newblock {\em Transactions on Pattern Analysis and Machine Intelligence},
  2009.

\bibitem{Baek2009}
J.~Baek, G.~McLachlan, and L.~Flack.
\newblock {Mixtures of Factor Analyzers with Common Factor Loadings:
  Applications to the Clustering and Visualisation of High-Dimensional Data}.
\newblock {\em IEEE Transactions on Pattern Analysis and Machine Intelligence},
  pages 1--13, 2009.

\bibitem{Bellman57}
R.~Bellman.
\newblock {\em {Dynamic Programming}}.
\newblock Princeton University Press, 1957.

\bibitem{Bibring2004}
J.-P. Bibring and 42~{co-authors}.
\newblock {Mars Surface Diversity as Revealed by the OMEGA/Mars Express
  Observations}.
\newblock {\em Science}, 307(5715):1576--1581, 2005.

\bibitem{Biernacki01}
C.~Biernacki, G.~Celeux, and G.~Govaert.
\newblock Assessing a mixture model for clustering with the integrated
  completed likelihood.
\newblock {\em IEEE Transactions on Pattern Analysis and Machine Intelligence},
  22(7):719--725, 2001.

\bibitem{Bouveyron12a}
C.~Bouveyron and C.~Brunet.
\newblock Simultaneous model-based clustering and visualization in the {Fisher}
  discriminative subspace.
\newblock {\em Statistics and Computing}, 22(1):301--324, 2012.

\bibitem{Bouveyron12}
C.~Bouveyron and C.~Brunet.
\newblock Theoretical and practical considerations on the convergence
  properties of the {Fisher-EM} algorithm.
\newblock {\em {J}ournal of {M}ultivariate {A}nalysis}, 109:29--41, 2012.

\bibitem{Bouveyron07}
C.~Bouveyron, S.~Girard, and C.~Schmid.
\newblock {High-Dimensional Data Clustering}.
\newblock {\em Computational Statistics and Data Analysis}, 52(1):502--519,
  2007.

\bibitem{Bouveyron07b}
C.~Bouveyron, S.~Girard, and C.~Schmid.
\newblock {High Dimensional Discriminant Analysis}.
\newblock {\em Communications in Statistics : Theory and Methods},
  36(14):2607--2623, 2007.

\bibitem{Cadima1995}
J.~Cadima and I.~Jolliffe.
\newblock Loadings and correlations in the interpretation of the principal
  components.
\newblock {\em Journal of Applied Statistics}, 22:203--214, 1995.

\bibitem{Celeux11}
G.~Celeux, M.-L. Martin-Magniette, C.~Maugis, and A.~Raftery.
\newblock Letter to the editor.
\newblock {\em Journal of the American Statistical Association}, 106(493),
  2011.

\bibitem{Efron2004}
B.~Efron, T.~Hastie, I.~Johnstone, and R.~Tibshirani.
\newblock {Least angle regression.}
\newblock {\em Annals of Statisics}, 32:407--499, May 2004.

\bibitem{Fisher36}
R.A. Fisher.
\newblock The use of multiple measurements in taxonomic problems.
\newblock {\em Annals of Eugenics}, 7:179--188, 1936.

\bibitem{Foley75}
D.H. Foley and J.W. Sammon.
\newblock An optimal set of discriminant vectors.
\newblock {\em IEEE Transactions on Computers}, 24:281--289, 1975.

\bibitem{Fukunaga90}
K.~Fukunaga.
\newblock {\em Introduction to Statistical Pattern Recognition}.
\newblock Academic. Press, San Diego, 1990.

\bibitem{Galimberti2009}
G.~Galimberti, A.~Montanari, and C.~Viroli.
\newblock {Penalized factor mixture analysis for variable selection in
  clustered data}.
\newblock {\em Computational Statistics \& Data Analysis}, 53(12):4301--4310,
  October 2009.

\bibitem{Ghahramani97}
Z.~Ghahramani and G.E. Hinton.
\newblock The {EM} algorithm for factor analyzers.
\newblock Technical report, University of Toronto, 1997.

\bibitem{Gower04}
J.C. Gower and G.B. Dijksterhuis.
\newblock Procrustes {P}roblems.
\newblock {\em Oxford {U}niversity {P}ress}, 2004.

\bibitem{Law04}
M.~Law, M.~Figueiredo, and A.~Jain.
\newblock {Simultaneous Feature Selection and Clustering Using Mixture Models}.
\newblock {\em IEEE Trans. on PAMI}, 26(9):1154--1166, 2004.

\bibitem{Liu03}
J.~Liu, J.L. Zhang, M.J. Palumbo, and C.E. Lawrence.
\newblock Bayesian clustering with variable and transformation selection.
\newblock {\em Bayesian Statistics}, 7:249--276, 2003.

\bibitem{Maugis09}
C.~Maugis, G.~Celeux, and M.-L. Martin-Magniette.
\newblock {Variable selection for Clustering with Gaussian Mixture Models}.
\newblock {\em Biometrics}, 65(3):701--709, 2009.

\bibitem{Maugis09a}
C.~Maugis, G.~Celeux, and M.-L. Martin-Magniette.
\newblock Variable selection in model-based clustering: A general variable role
  modeling.
\newblock {\em Computational Statistics and Data Analysis}, 53:3872--3882,
  2009.

\bibitem{McLachlan2003}
G.~McLachlan, D.~Peel, and R.~Bean.
\newblock Modelling high-dimensional data by mixtures of factor analyzers.
\newblock {\em Computational Statistics and Data Analysis}, (41):379, 2003.

\bibitem{McNicholas2008}
P.~McNicholas and B.~Murphy.
\newblock {Parsimonious Gaussian mixture models}.
\newblock {\em Statistics and Computing}, 18(3):285--296, 2008.

\bibitem{Montanari06}
A.~Montanari and C.~Viroli.
\newblock Dimensionally reduced mixtures of regression models.
\newblock {\em Electronic Proceedings of KNEMO, Knowledge Extraction and
  Modelling}, 2006.

\bibitem{Montanari2010}
A.~Montanari and C.~Viroli.
\newblock {Heteroscedastic Factor Mixture Analysis}.
\newblock {\em Statistical Modeling: An International journal}, 10(4):441--460,
  2010.

\bibitem{Pan07}
W.~Pan and X.~Shen.
\newblock Penalized model-based clustering with application to variable
  selection.
\newblock {\em Journal of Machine Learning Research}, 8:1145--1164, 2007.

\bibitem{Qiao09}
Z.~Qiao, L.~Zhou, and J.Z. Huang.
\newblock Sparse linear discriminant analysis with applications to high
  dimensional low sample size data.
\newblock {\em International Journal of Applied Mathematics}, 39(1), 2009.

\bibitem{Raftery06}
A.~Raftery and N.~Dean.
\newblock Variable selection for model-based clustering.
\newblock {\em Journal of the American Statistical Association},
  101(473):168--178, 2006.

\bibitem{Tibshirani01}
R.~Tibshirani, G.~Walther, and T.~Hastie.
\newblock Estimating the number of clusters in a dataset via the gap statistic.
\newblock {\em Journal of the Royal Statistical Society, Series B},
  32(2):411--423, 2001.

\bibitem{Tipping99b}
E.~Tipping and C.~Bishop.
\newblock {Mixtures of Probabilistic Principal Component Analysers}.
\newblock {\em Neural Computation}, 11(2):443--482, 1999.

\bibitem{Wang08}
S.~Wang and J.~Zhou.
\newblock Variable selection for model-based high dimensional clustering and
  its application to microarray data.
\newblock {\em Biometrics}, 64:440--448, 2008.

\bibitem{Witten10}
D.M. Witten and R.~Tibshirani.
\newblock A framework for feature selection in clustering.
\newblock {\em Journal of the American Statistical Association},
  105(490):713--726, 2010.

\bibitem{Witten09}
D.M. Witten, R.~Tibshirani, and T.~Hastie.
\newblock A penalized matrix decomposition, with applications to sparse
  principal components and canonical correlation analysis.
\newblock {\em Biostatistic}, 10(3):515--534, 2009.

\bibitem{Xie08}
B.~Xie, W.~Pan, and X.~Shen.
\newblock Penalized model-based clustering with cluster-specific diagonal
  covariance matrices and grouped variables.
\newblock {\em Electrical Journal of Statistics}, 2:168--212, 2008.

\bibitem{Xie10}
B.~Xie, W.~Pan, and X.~Shen.
\newblock Penalized mixtures of factor analyzers with application to clustering
  high-dimensional microarray data.
\newblock {\em Bioinformatics}, 26(4):501--508, 2010.

\bibitem{Yoshida04}
R.~Yoshida, T.~Higuchi, and S.~Imoto.
\newblock A mixed factor model for dimension reduction and extraction of a
  group structure in gene expression data.
\newblock {\em IEEE Computational Systems Bioinformatics Conference},
  8:161--172, 2004.

\bibitem{Yoshida06}
R.~Yoshida, T.~Higuchi, S.~Imoto, and S.~Miyano.
\newblock Array cluster: an analytic tool for clustering, data visualization
  and model finder on gene expression profiles.
\newblock {\em Bioinformatics}, 22:1538--1539, 2006.

\bibitem{Zhang09}
Z.~Zhang, G.~Dai, and M.I. Jordan.
\newblock A flexible and efficient algorithm for regularized fisher
  discriminant analysis.
\newblock In {\em Proceedings of the European Conference on Machine Learning
  and Knowledge Discovery in Databases}, pages 632--647, 2009.

\bibitem{Zou2006}
H.~Zou and R.~Hastie, T.and~Tibshirani.
\newblock {Sparse Principal Component Analysis}.
\newblock {\em Journal of Computational and Graphical Statistics},
  15(2):265--286, June 2006.

\bibitem{Zou03}
H.~Zou and T.~Hastie.
\newblock Regularization and variable selection via the elastic net.
\newblock {\em Journal of the Royal Statistical Society}, 67:301--320, 2005.

\bibitem{Zou07}
H.~Zou, T.~Hastie, and R.~Tibshirani.
\newblock On the degrees of freedom of the {L}asso.
\newblock {\em Annals of Statistics}, 35(5):2173--2192, 2007.

\end{thebibliography}

\end{document}